\documentclass[a4paper,UKenglish,cleveref,thm-restate]{lipics-v2021}

\hideLIPIcs  

\graphicspath{{./figures/}}

\usepackage{charter,eulervm}
\usepackage{bbding}
\usepackage{eulervm}
\usepackage[numbers]{natbib}
\usepackage{parskip}
\usepackage{float}
\usepackage{hyperref}

\usepackage{multirow}

\title{On the Complexity of the ~~~~~~~~~~~~~~~~~~~Eigenvalue Deletion Problem} 

\titlerunning{On the Complexity of the Eigenvalue Deletion Problem}

\author{Neeldhara Misra}{Indian Institute of Technology, Gandhinagar \and \url{https://www.neeldhara.com}}{neeldhara.m@iitgn.ac.in}{https://orcid.org/0000-0003-1727-5388}{Supported by DST-SERB and IIT Gandhinagar.}

\author{Harshil Mittal}{Indian Institute of Technology, Gandhinagar}{mittal_harshil@iitgn.ac.in}{}{Supported by IIT Gandhinagar.}

\author{Saket Saurabh}{Institute of Mathematical Sciences, Chennai \and University of Bergen \and \url{https://sites.google.com/view/sakethome}}{saket@imsc.res.in}{https://orcid.org/0000-0001-7847-6402}{Supported by ERC, the University of Bergen, and IMSc.}

\author{Dhara Thakkar}{Indian Institute of Technology, Gandhinagar \and \url{https://sites.google.com/iitgn.ac.in/dharathakkar}}{thakkar_dhara@iitgn.ac.in}{https://orcid.org/0000-0002-4234-0105}{Supported by CSIR-UGC NET JRF Fellowship.}

\authorrunning{N. Misra, H. Mittal, S. Saurabh, and D. Thakkar} 

\Copyright{Jane Open Access and Joan R. Public} 
    


\category{} 

\relatedversion{} 

\acknowledgements{We thank Daniel Lokshtanov for helpful discussions.}

\nolinenumbers 

\EventEditors{John Q. Open and Joan R. Access}
\EventNoEds{2}
\EventLongTitle{42nd Conference on Very Important Topics (CVIT 2016)}
\EventShortTitle{CVIT 2016}
\EventAcronym{CVIT}
\EventYear{2016}
\EventDate{December 24--27, 2016}
\EventLocation{Little Whinging, United Kingdom}
\EventLogo{}
\SeriesVolume{42}
\ArticleNo{23}

\begin{document}

\maketitle

\begin{abstract}
    For any fixed positive integer $r$ and a given budget $k$, the $r$-\textsc{Eigenvalue Vertex Deletion} ($r$-EVD) problem asks if a graph $G$ admits a subset $S$ of at most $k$ vertices such that the adjacency matrix of $G\setminus S$ has at most $r$ distinct eigenvalues. The edge deletion, edge addition, and edge editing variants are defined analogously. For $r = 1$, $r$-EVD is equivalent to the Vertex Cover problem. For $r = 2$, it turns out that $r$-EVD amounts to removing a subset $S$ of at most $k$ vertices so that $G\setminus S$ is a cluster graph where all connected components have the same size.

    We show that $r$-EVD is \textsf{NP}-complete even on bipartite graphs with maximum degree four for every fixed $r > 2$, and \textsf{FPT} when parameterized by the solution size and the maximum degree of the graph.

    We also establish several results for the special case when $r = 2$. For the vertex deletion variant, we show that $2$-EVD is \textsf{NP}-complete even on triangle-free and $3d$-regular graphs for any $d\geq 2$, and also \textsf{NP}-complete on $d$-regular graphs for any $d\geq 8$. The edge deletion, addition, and editing variants are all \textsf{NP}-complete for $r = 2$. The edge deletion problem admits a polynomial time algorithm if the input is a cluster graph, while --- in contrast --- the edge addition variant is hard even when the input is a cluster graph. We show that the edge addition variant has a quadratic kernel. The edge deletion and vertex deletion variants admit a single-exponential \textsf{FPT} algorithm when parameterized by the solution size alone.
    
    Our main contribution is to develop the complexity landscape for the problem of modifying a graph with the aim of reducing the number of distinct eigenvalues in the spectrum of its adjacency matrix. It turns out that this captures, apart from Vertex Cover, also a natural variation of the problem of modifying to a cluster graph as a special case, which we believe may be of independent interest.    
\end{abstract}

\newpage

\section{Introduction}

Graph modification problems are a fundamental class of optimization problems where we have a class of graphs $\mathcal F$ that satisfy some property of interest $P$, the input is a graph $G$, and we are interested in a smallest subset of vertices $S \subseteq V(G)$ such that $G \setminus S \in \mathcal{F}$. This is a rather general framework that captures several classical optimization problems as special cases, for instance:

\begin{itemize}
    \item when $\mathcal{F}$ is the collection of edgeless graphs, then the problem is \textsc{Vertex Cover};
    \item when $\mathcal{F}$ is the collection of acyclic graphs, then the problem is \textsc{Feedback Vertex Set};
    \item when $\mathcal{F}$ is the collection of bipartite graphs, then the problem is \textsc{Odd Cycle Traversal};
\end{itemize}

and so on. It has also been of interest to study modifications other than vertex deletion: the most common alternate modifications considered include edge deletion, edge addition, and edge editing (adding and removing edges). The optimization problems for these operations may be posed analogously.

\citeauthor{MMS2016} pose the question of modifying a graph with the goal of reducing the rank of the associated adjacency matrix, which is to say that $\mathcal{F}_{\leqslant r}$ is the class of graphs whose adjacency matrices have rank at most $r$. We use $A_G$ to denote the adjacency matrix of a graph $G$, and we use the phrase ``spectrum of $G$'' to refer to the (multi-)set of eigenvalues of $A_G$. Previous works focus separately on the settings of undirected~\citep{MMS2016} and directed~\citep{MS2018} graphs.

In the setting of simple undirected graphs, \citeauthor{MMS2016} introduce and study the \textsc{$r$-Rank Vertex Deletion}, \textsc{$r$-Rank Edge Deletion}, and \textsc{$r$-Rank Editing} problems. These problems generalize the classical \textsc{Vertex Cover} problem. They show that all the three problems are \textsf{NP}-complete, and are fixed parameter tractable (\textsf{FPT}) in the standard parameter: in particular, they demonstrate an algorithm with running time $2^{\mathcal{O}(k \log r)} n^{\mathcal{O}(1)}$ for $r$-\textsc{Rank Vertex Deletion}, and an algorithm for $r$-\textsc{Rank Edge Deletion} and $r$-\textsc{Rank Editing} running in time $2^{\mathcal{O}(f(r) \sqrt{k} \log k)} n^{\mathcal{O}(1)}$, where $k$ is the size of the solution sought. The authors also leave the following question open:

\emph{``[$\ldots$] what is complexity of the problem of reducing the number of distinct eigenvalues of a graph by deleting a few vertices or editing a few edges?''}

In this paper, we address this question at length, developing an initial picture of the complexity landscape for what we call the $r$-\textsc{Eigenvalue Vertex Deletion} ($r$-EVD), $r$-\textsc{Eigenvalue Edge Deletion} ($r$-EED), $r$-\textsc{Eigenvalue Edge Addition} ($r$-EEA), and $r$-\textsc{Eigenvalue Edge Editing} ($r$-EEE) problems. All these problems are defined for an arbitrary but fixed positive integer $r$. 


The problem definitions are the following, where we are given an undirected graph $G$ and a positive integer $k$ as input in all cases:

\begin{itemize}
\item $\mathbf{r}$\textbf{-EVD.} Is there a set $S\subseteq V(G)$ of size $\leq k$ such that the number of distinct eigenvalues of $A_{G \setminus S}$ is at most $r$?

\item $\mathbf{r}$\textbf{-EEE.} Is there a set $F \subseteq \binom{V(G)}{2}$ of size $\leqslant k$ such that the number of distinct eigenvalues of $A_{H}$ is at most $r$, where $H:=(V(G),E(G) \Delta F)$?

\item $\mathbf{r}$\textbf{-EEA.} Is there a set $F \subseteq {V(G) \choose 2} \setminus E(G)$ of size $\leqslant k$ such that the number of distinct eigenvalues of $A_{H}$ is at most $r$, where $H:=(V(G),E(G) \cup F)$?

\item $\mathbf{r}$\textbf{-EED.} Is there a set $F \subseteq E(G)$ of size $\leqslant k$ such that the number of distinct eigenvalues of $A_{H}$ is at most $r$, where $H:=(V(G),E(G) \setminus F)$?
\end{itemize}

Note that if we have a solution $S$ for the $r$-\textsc{Rank Vertex Deletion} problem, then $S$ is also a solution for the $(r+1)$-\textsc{Eigenvalue Vertex Deletion} problem; and analogous statements hold for the other modification problems. This is because the $r$-\textsc{Rank Vertex Deletion} problem can be equivalently stated as follows: given a graph $G$ and a positive integer $k$, find a smallest subset of vertices $S \subseteq V(G)$ such that $G \setminus S$ has at most $r$ non-zero eigenvalues. However, the converse is not true (since, in general, bounding the number of distinct eigenvalues is not sufficient to bound the rank), making the eigenvalue deletion problems distinct from their rank deletion counterpart.

\bgroup
\def\arraystretch{1.25}
\begin{table}[t]
  \centering
  \resizebox{\textwidth}{!}{%
  \begin{tabular}{lll}
  \hline
                                      & $r = 2$                                                                                                     & Fixed $r \geqslant 3$                                                                                     \\ \hline
  \multirow{3}{*}{Vertex Deletion~~~} & \begin{tabular}[c]{@{}l@{}}\textsf{NP}-complete for d-regular graphs~$\dagger$\\ (\Cref{2EVD hard})\end{tabular}       & \begin{tabular}[c]{@{}l@{}}\textsf{NP}-complete even on bipartite graphs\\ (\Cref{rEVD hard})\end{tabular}                \\ \cline{2-3} 
                                      & \begin{tabular}[c]{@{}l@{}}\textsf{FPT} in $k$ (\Cref{2EVD fpt})\end{tabular}                                      & \multirow{2}{*}{\begin{tabular}[c]{@{}l@{}}\textsf{FPT} in $k$ and $\Delta(G)$\\ (\Cref{rEVD FPT})\end{tabular}}   \\ \cline{2-2}
                                      & \begin{tabular}[c]{@{}l@{}}Polynomial time on forests\\ (\Cref{2EVD trees})\end{tabular}            &                                                                                                           \\ \hline
  \multirow{2}{*}{Edge Addition}      & \begin{tabular}[c]{@{}l@{}}\textsf{NP}-complete even on cluster graphs\\ (\Cref{2EEA NPhard})\end{tabular}   & \multirow{2}{*}{NP-complete~ (\Cref{rEEA NP hardness})}                                                                                     \\ \cline{2-2}
                                      & \begin{tabular}[c]{@{}l@{}}Quadratic kernel in $k$ (\Cref{2EEA kernel})\end{tabular}                      &                                                                                                           \\ \hline
  \multirow{2}{*}{Edge Deletion}      & \begin{tabular}[c]{@{}l@{}}\textsf{NP-complete}  (\Cref{2EEE NP hard})\end{tabular}                                      & \multirow{2}{*}{\begin{tabular}[c]{@{}l@{}}\textsf{NP}-complete\\ (\Cref{rEED NP hard})\end{tabular}} \\ \cline{2-2}
                                      & \begin{tabular}[c]{@{}l@{}}\textsf{FPT} in $k$ (\Cref{2EED fpt})\end{tabular}
                                      \\ \cline{2-2}
                                      & \begin{tabular}[c]{@{}l@{}}Polynomial time on triangle-free graphs\\ (\Cref{EED trianglefree})\end{tabular}
                                      &                                                                                                           \\ \hline
  Edge Editing                        & \begin{tabular}[c]{@{}l@{}}\textsf{NP}-complete (\Cref{2EEE NP hard})\end{tabular}                                 &  OPEN                                                                                                         \\ ~\\
  \end{tabular}%
  }
  \caption{A summary of our results. The result marked $\dagger$ holds for all $d$ except for $d = 1,2,3,4,5,7$. Some polynomial cases are omitted from this summary.
  \vspace{-20pt}}
  \label{table:resultsummary}
  \end{table}
\egroup

\textbf{Our Contributions.} We summarize our contributions below, and also in~\Cref{table:resultsummary}. We first focus on the special case when $r = 2$. It is known that the adjacency matrix $A_{G}$ of a graph $G$ has at most two distinct eigenvalues if and only if $G$ is a disjoint union of equal-sized cliques (\Cref{2eval}). Based on this, note that the $2$-\textsc{Eigenvalue Vertex Deletion} problem is equivalent to finding a subset $S \subseteq V(G)$ of vertices such that $G \setminus S$ is a disjoint union of cliques of size $\ell$ for some $1 \leq \ell \leq |V(G)|$. Note that this is closely related to the \textsc{Cluster Vertex Deletion} problem, which is a well-studied question that involves removing a smallest subset of vertices to obtain a cluster graph. However, to the best of our knowledge, the variant where we demand that the clusters have the same size has not been studied. Our results about the ``uniform'' version of \textsc{Cluster Vertex Deletion} may therefore be of independent interest. 

Our main contributions in the context of vertex deletion are the following results:

\begin{itemize}
    \item We show that $2$-EVD is \textsf{NP}-complete on $d$-regular graphs for all $d$ except for $d = 1,2,3,4,5,7$ (\Cref{2EVD hard}).
    \item We also give a single-exponential \textsf{FPT} algorithm in the standard parameter (\Cref{2EVD fpt}), and show that the problem can be solved in polynomial time on forests and $d$-regular graphs for $d \leqslant 2$ (\Cref{2EVD trees}).
    \item Further, for any fixed $r \geqslant 3$, we show that $r$-EVD is \textsf{NP}-complete on bipartite graphs (\Cref{rEVD hard}) and is \textsf{FPT} in the standard parameter combined with the maximum degree of the graph (\Cref{rEVD FPT}). 
\end{itemize}



We now describe our findings for the edge modification variants. 
\begin{itemize}
    \item We show that $2$-EEA is already \textsf{NP}-complete when the input is either a cluster graph, a forest, or a collection of cycles (\Cref{2EEA NPhard}). \item We demonstrate that the problem has a quadratic kernel in the standard parameter (\Cref{2EEA kernel}).
    \item We show that $r$-\textsc{EEA} is NP-complete for any fixed $r\geq 3$ (Theorem~\ref{rEEA NP hardness}).
    \item For the edge deletion variant, we show that $r$-EED is \textsf{NP}-complete for any fixed $r \geqslant 2$ (\Cref{2EEE NP hard,rEED NP hard}). 
    \item For $2$-EED, we have a single-exponential \textsf{FPT} algorithm (\Cref{2EED fpt}) in the standard parameter and a polynomial time algorithm on triangle-free graphs (\Cref{EED trianglefree}). 
    \item Finally, for the edge editing variant, we show that $2$-\textsc{Eigenvalue Edge Editing} is \textsf{NP}-complete (\Cref{2EEE NP hard}). 
\end{itemize}

\textbf{Related Work.} As we noted previously, the special case when $r=2$ is closely related to the problem of modifying to a cluster graph, in which we are allowed to modify the graph such that the resulting graph is cluster i.e., it is disjoint union of cliques. Depending on the modifications allowed, these problems are variously refered to as \textsc{Cluster Vertex Deletion}, \textsc{Cluster Edge Deletion}, \textsc{Cluster Edge Addition} and \textsc{Cluster Edge Editing}. Further, \citeauthor{CVD-Modification} have studied a variant of cluster vertex deletion where they additionally demand that the cluster graph obtained after the modification has at most $p$ components.


Problems related to modifying to a cluster graph are very well-studied because they model the clustering problem in various ways. In a clustering problem we are given various data points with some notion of distance between these points, and it is of interest to group these points into ``clusters'', where each cluster consists of points that are mutually close with respect to the given distance metric. These scenarios can often be modeled with graphs, and in fact graph structure can often be used to model additional constraints of interest. Given the fundamental importance of clustering, it is no surprise that modifying to cluster graphs has attracted substantial interest in the literature of graph algorithms. We refer the reader to~\cite{CVD-master-thesis} for an overview of results related to cluster modification problems.

Another related problem is the problem of deleting to a graph where the connected components have small diameter. This is known as the $s$-\textsc{Club Cluster Vertex Deletion} problem \citep{sClubVD}. Here, we are given a graph $G$ and two integers $s \geq 2$ and $k \geq 1$; and the question is if it is possible to remove at most $k$ vertices from $G$ such that each connected component of the resulting graph has diameter at most $s$. Note that this naturally generalizes the problem of modifying to cluster graphs: indeed, the problem is equivalent to \textsc{Cluster Vertex Deletion} for $s = 1$. The edge modification variants have also been considered and are well-studied. 

We note that a solution to the $r$-\textsc{Eigenvalue Vertex Deletion}  problem will also be a valid solution to the $(r-1)$-\textsc{Club Cluster Vertex Deletion} due to~\Cref{diameter bound}, which states that graphs of diameter $d$ have at least $(d+1)$ distinct eigenvalues. This is analogously true for the other modification problems as well. On the other hand, it is easy to see that the converse is not necessarily true. 

Throughout, we use the $\mathcal{O}^\star(\cdot)$ notation to suppress polynomial factors. 

Sections 3,4,5, and 6 focus respectively on the problems of $r$-EVD, $r$-EEA, $r$-EED, and $r$-EEE.



\section{Preliminaries}
Let $G=(V,E)$ be a graph, where $V$ and $E$ denote the vertex set and the edge set of $G$ respectively. We typically use $n$ and $m$ to denote $|V|$ and $|E|$ respectively. Throughout this paper, we focus on simple and undirected graphs. The adjacency matrix  $A_{G}=a_{ij}$ of a graph $G$ is an $n\times n$ matrix with $a_{ij}\in \{0,1\}$ the entry $(i.j)=1$ if the pair $(i,j)$ is an edge in $G$. The \emph{spectrum} of $G$ is the multi-set of eigenvalues of $A_{G}$. We note that spectrum can be computed in polynomial time based on results from \cite{eigenvalueComputationIso} and \cite{eigenvalueComputationSymMat}. Notice that for simple undirected graphs $G$, $A_{G}$ is symmetric matrix with zero on the diagonals.

A \emph{principal submatrix} of a square matrix $A$ is a matrix obtained by removing an equal number of rows and columns from $A$ such that the indices of the removed rows match with the indices of the removed columns.

The following known results will be relevant to our discussions:

\begin{lemma}
[\cite{goldberg2014split,TwoDistinctEV}]
\label{2eval} Let $G$ be a graph. Then, its adjacency matrix $A_{G}$ has at most two distinct eigenvalues if and only if $G$ is a disjoint union of equal-sized cliques. 
\end{lemma}

\begin{lemma}
[\cite{brouwer2011spectra}, Proposition 1.3.3] 
\label{diameter bound}
Let $G$ be a connected graph with diameter $d$. Then, its adjacency matrix $A_G$ has at least $d+1$ distinct eigenvalues.
\end{lemma}

\begin{lemma}
[\cite{brouwer2011spectra}, Corollary 2.5.2]
\label{Cauchy interlacing}\textbf{Cauchy interlacing}.\\Let $A$ be a symmetric matrix of size $n\times n$. Let $B$ be a principal submatrix of $A$ of size $(n-1)\times (n-1)$. Then, the eigenvalues of $B$ interlace the eigenvalues of $A$. That is,
\[\mu_1\geq \sigma_1\geq \mu_2\geq \sigma_2\geq \mu_3\geq \ldots \ldots\ldots\ldots\geq \mu_{n-2}\geq\sigma_{n-2}\geq\mu_{n-1}\geq \sigma_{n-1}\geq \mu_{n}\]
where, $\mu_1\geq \mu_2\geq\ldots\ldots\geq\mu_{n}$ denote the $n$ eigenvalues of $A$, and $\sigma_1\geq \sigma_{2}\geq \ldots\ldots \geq \sigma_{n-1}$ denote the $n-1$ eigenvalues of $B$.
\end{lemma}

\begin{lemma}
[\cite{brouwer2011spectra}, Chapter 3, Exercise 1]
\label{smallest ev -1} Let $G$ be a graph with smallest eigenvalue $-1$. Then, $G$ is a disjoint union of cliques.
\end{lemma}

Some examples of graph classes whose spectrum is well-known include complete graphs, paths and cycles (\cite{brouwer2011spectra}, Chapter 1). A complete graph on $n$ vertices has eigenvalues $-1$ and $n-1$ (with multiplicities $n-1$ and $1$ respectively). A path on $n$ vertices has eigenvalues $2\cos\big(\frac{\pi j}{n+1}\big)\big\lvert_{1\leq j\leq n}$. A cycle on $n$ vertices has eigenvalues $2\cos\big(\frac{2\pi j}{n}\big)\big\lvert_{0\leq j\leq n-1}$.

We refer the reader to~\cite{diestel} for background on graph theory and as a reference for standard graph-theoretic notation, and~\cite{pcbook} for background on terminology related to parameterized algorithms. In particular, the machinery we use for our branching algorithms is based on the ideas described in~\cite[Chapter 3][]{pcbook}.

In the following, let $G=\left(V,E\right)$ and $G'=\left(V',E'\right)$ be graphs, and $U \subseteq V$ some subset of vertices of $G$. Let $G'$ be a subgraph of $G$. 
If $E'$ contains all the edges $\left\{ u,v\right\} \in E$ with $u,v\in V'$, then $G'$ is an \emph{induced subgraph} of $G$, \emph{induced by} $V'$, denoted by $G[V']$.  
For any $U\subseteq V$, $G\setminus U=G[V\setminus U]$.

Two vertices $x, y$ of $G$ are adjacent, or neighbors, if $\{x,y\}$ is an edge of $G$. The degree of a vertex $v$, denoted $d(v)$, is the number of nieghbors it has. A graph is $d$-regular if every vertex has degree $d$. If all the vertices of $G$ are pairwise adjacent, then $G$ is complete. For $v\in V$, $N_G(v)=\{u\;|\; (u,v)\in E\}$: this set collects the neighbors of $v$ and is called the \emph{open neighborhood} of $v$, and we may drop the subscript $G$ if the graph is clear from the context. For $v\in V$, $N_G[v]= N_G(v) \cup \{v\}$: this set is the \emph{closed neighborhood} of $v$.

Pairwise non-adjacent vertices or edges are called independent. More formally, a set of vertices or of edges is independent (or stable) if no two of its elements are adjacent. If $S \subseteq V(G)$ is an independent set, then the component of $S$ in $G$ is called a vertex cover.

A path is a non-empty graph $P=(V, E)$ of the form $$
V=\left\{x_0, x_1, \ldots, x_k\right\} \quad E=\left\{x_0 x_1, x_1 x_2, \ldots, x_{k-1} x_k\right\},$$ where the $x_i$ are all distinct. The vertices $x_0$ and $x_k$ are linked by $P$ and are called its ends; the vertices $x_1, \ldots, x_{k-1}$ are the inner vertices of $P$. The number of edges of a path is its length, and the path of length $k$ is denoted by $P^k$. If $P=x_0 \ldots x_{k-1}$ is a path and $k \geqslant 3$, then the graph $C:=$ $P+x_{k-1} x_0$ is called a cycle. The distance $d_G(x, y)$ in $G$ of two vertices $x, y$ is the length of a shortest $x-y$ path in $G$; if no such path exists, we set $d(x, y):=\infty$. The greatest distance between any two vertices in $G$ is the \emph{diameter} of $G$, denoted by $\operatorname{d}(G)$. A non-empty graph $G$ is called \emph{connected} if any two of its vertices are linked by a path in $G$. If $U \subseteq V(G)$ and $G[U]$ is connected, we also call $U$ itself connected (in $G$ ). A maximal connected subgraph of $G$ is a connected component of $G$. A graph $G$ is called \emph{cluster} graph if every connected component induces a clique. A graph is a \emph{forest} if every connected component does not contain a cycle as a subgraph.

\section{Reducing eigenvalues by deleting vertices}
In this section, we show that the $r$-EVD problem is \textsf{NP}-complete for $r \geq 1$. Recall that for $r = 1$, $r$-EVD is equivalent to \textsc{Vertex Cover}. For $r = 2$, we show that the problem is \textsf{NP}-complete on general graphs, admits a single-exponential \textsf{FPT} algorithm in the standard parameter, and is polynomial-time solvable on trees. For any fixed $r \geqslant 3$, we show that the problem is \textsf{NP}-complete on bipartite graphs and is \textsf{FPT} in the standard parameter combined with the maximum degree of the graph. 

 

\subsection{Deleting to Two Distinct Eigenvalues}

Note that by~\Cref{2eval}, $2$-EVD is equivalent to \textsc{Uniform Cluster Vertex Deletion}, a problem where the input is a graph $G$ and a positive integer $k$ and the question is if there is a subset $S \subseteq V(G)$ of vertices such that $G \setminus S$ is a disjoint union of $\ell$-sized cliques for some $1 \leq \ell \leq |V(G)|$. Note that $\ell$ is not a part of the input. We begin by showing that the problem is hard even when restricted to $d$-regular graphs for any $d$ other than $1,2,3,4,5,7$.

\begin{restatable}{theorem}{twoEVDhard}
  \label{2EVD hard} 
    $2$-\textsc{Eigenvalue Vertex Deletion} is \textsf{NP}-complete even on triangle-free and $3d$-regular graphs for any $d\geq 2$, and \textsf{NP}-complete on $d$-regular graphs for any $d\geq 8$.
\end{restatable} 
  
To show this result we use two reductions: one from the \textsc{Independent Set} problem on cubic triangle-free graphs and the other from \textsc{Independent Set} on planar cubic triangle-free graphs. 

In the first construction, we replace every vertex $v$ with vertices $v^{(1)}$ and $v^{(2)}$, and extended the edges as follows: an edge $(u,v)$ maps to the edges $(u^{(1)},v^{(1)})$, $(u^{(1)},v^{(2)})$, $(u^{(2)},v^{(1)})$, and $(u^{(2)},v^{(2)})$. Note that this construction preserves triangle-freeness and transforms a cubic graph to a six-regular graph. For demonstrating hardness on $3d$ regular graphs for $d \geqslant 2$, we make $d$ copies of the vertices instead of two copies. 

For the second construction, we make six copies of the graph and for every vertex, we induce a clique on all its copies. This construction turns a cubic graph into a $8$-regular graph. For demonstrating hardness on $d$ regular graphs for $d \geqslant 8$, we make $(d-2)$ copies of the vertices instead of six. 

\begin{proof}$~$\\
\textbf{First reduction.}\\
Consider an instance, say $(G,z)$, of \textsc{Independent Set}, where $G$ is a cubic triangle-free graph, say on $n$ vertices. Construct a graph, say $H$, as follows: For each vertex $v\in V(G)$, introduce two copies of $v$, say $v^{(1)}$ and $v^{(2)}$. Also, for each edge $e\in E(G)$, say with endpoints $u$ and $v$, make the $i^{th}$ copy of $u$, i.e., $u^{(i)}$, adjacent to the $j^{th}$ copy of $v$, i.e., $v^{(j)}$, for all $1\leq i,j\leq 2$.
  
  \begin{center}
  \includegraphics[scale=0.65]{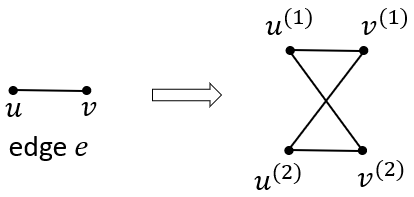}
  \end{center}

  That is, 
  $$V(H) = \big\{v^{(i)}~|~v\in V(G)~\mbox{and}~1\leq i\leq 2\big\}$$
  $$E(H) = \big\{\{u^{(i)},v^{(j)}\}~|~\{u,v\}\in E(G)~\mbox{and}~1\leq i,j \leq 2\big\}$$
  
 See \Cref{first reduction} for an example.
\begin{figure}[H]
\centering
\includegraphics[scale=0.67]{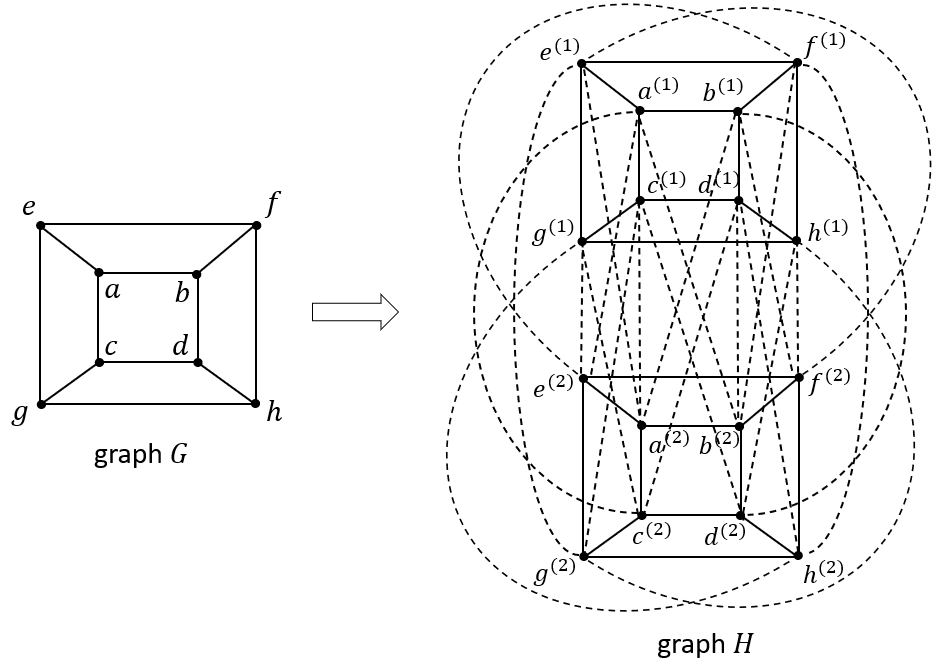}
\caption{An example illustrating the construction of $H$ from $G$ in the first reduction of \Cref{2EVD hard}.}
\label{first reduction}
\end{figure}

Note that for every vertex $v\in V(G)$, each of the two 
  copies of $v$, i.e., $v^{(1)}$ and $v^{(2)}$, has six neighbours in $H$, namely the two copies of each of the three neighbours of $v$ in $G$. So, $H$ is a $6$-regular graph. Also, as $G$ is a triangle-free graph, so is $H$. 
  
  Let us show that $G$ has an independent set of size $z$ if and only if $\big(H,2(n-z)\big)$ is a \textbf{\textsc{YES}} instance of $2$-\textsc{Eigenvalue Vertex Deletion}.
  
  \noindent ($\Rightarrow$) Suppose that $G$ has an independent set, say $I$, of size $z$. Let $I'\subseteq V(H)$ denote the set that consists of both copies of each vertex in $I$. That is, $$I' = \big\{v^{(i)}~|~v\in I~\mbox{and}~1\leq i\leq 2\big\}$$
  Note that $I'$ is an independent set in $H$. So, all eigenvalues of the adjacency matrix of $H[I']$ are $0$. Thus, as $|V(H)\setminus I'|=2(n-z)$, it follows that $\big(H,2(n-z)\big)$ is a \textbf{\textsc{YES}} instance of $2$-\textsc{Eigenvalue Vertex Deletion}.
  
  \noindent ($\Leftarrow$) Suppose that $\big(H,2(n-z)\big)$ is a \textbf{\textsc{YES}} instance of $2$-\textsc{Eigenvalue Vertex Deletion}. That is, there exists $S\subseteq V(H)$ of size $\leq 2(n-z)$ such that the adjacency matrix of $H\setminus S$ has at most two distinct eigenvalues. Using Lemma~\ref{2eval}, $H\setminus S$ is a disjoint union of equal-sized cliques, say $C_1,\ldots, C_t$, each of size $s$. We have $|V(H)\setminus S| = s\cdot t \geq 2z$. As $H$ has no triangles, we get $s\leq 2$.
  
  \noindent Case 1: $s=1$\\
  For each $1\leq i\leq t$, let $u_i^{(\alpha_i)}$ denote the vertex of $C_i$. That is, the clique $C_i$ consists of the $\alpha_i^{th}$ copy of the vertex $u_i$ of $G$. Note that for any $1\leq i<j\leq t$, $u_i$ is not adjacent to $u_j$ in $G$; otherwise, there's an edge joining the cliques $C_i$ and $C_j$, namely $\big\{u_i^{(\alpha_i)}, u_j^{(\alpha_j)}\big\}$, as shown below.
  \begin{center}
  \includegraphics[scale=0.67]{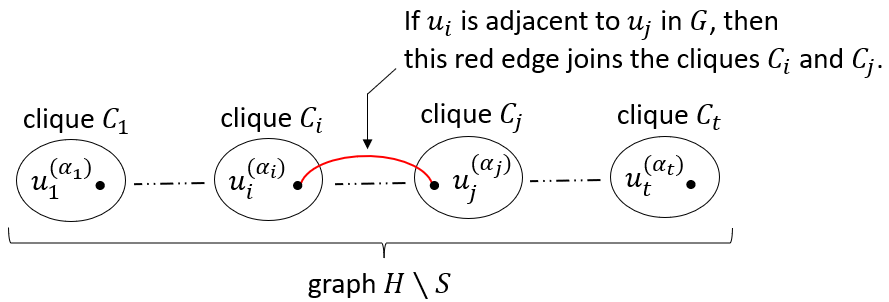}
  \end{center}
  Also, as any vertex of $G$ appears at most twice in $u_1,\ldots, u_t$, there are at least $\frac{t}{2}\geq z$ distinct vertices therein. Therefore, $G$ has an independent set of size $z$.

  \noindent Case 2: $s=2$\\
  For each $1\leq i\leq t$, let $u_i^{(\alpha_i)}$ and $v_i^{(\beta_i)}$ denote the two vertices of $C_i$. That is, the clique $C_i$ consists of the $\alpha_i^{th}$ copy of the vertex $u_i$ and the $\beta_i^{th}$ copy of the vertex $v_i$. Consider any $1\leq i<j\leq t$. Note that $u_i$ is not adjacent to $u_j$ in $G$; otherwise, there's an edge joining the cliques $C_i$ and $C_j$, namely $\big\{u_i^{(\alpha_i)}, u_j^{(\alpha_j)}\big\}$, as shown below.
  \begin{center}
  \includegraphics[scale=0.67]{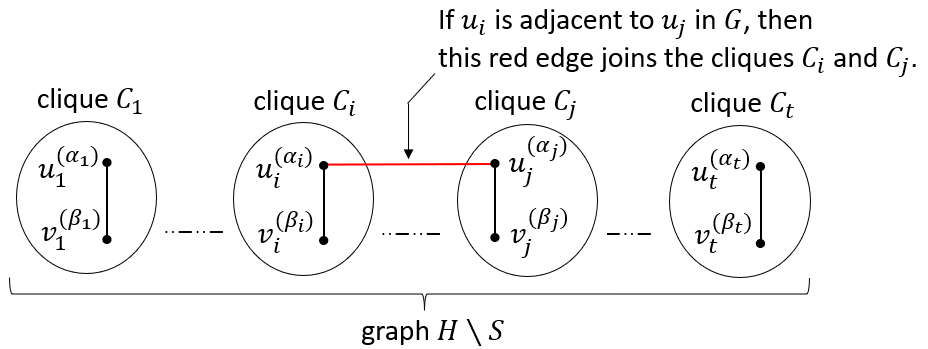}
  \end{center}
  Next, let us show that $u_i$ is distinct from $u_j$. For the sake of contradiction, assume that $u_i$ and $u_j$ are the same vertex, say $u$, of $G$. Without loss of generality, $C_i$ contains the first copy of $u$, i.e., $u^{(1)}$, and $C_j$ contains the second copy of $u$, i.e., $u^{(2)}$. Observe that $u^{(1)}$ and $u^{(2)}$ are twins in $H$. Thus, as $u^{(1)}$ is adjacent to $v_i^{(\beta_i)}$, so is $u^{(2)}$. Likewise, as $u^{(2)}$ is adjacent to $v_j^{(\beta_j)}$, so is $u^{(1)}$. Hence, as shown below, there are two edges joining the cliques $C_i$ and $C_j$, namely $\big\{u^{(1)}, v_j^{(\beta_j)}\big\}$ and $\big\{u^{(2)}, v_i^{(\beta_i)}\big\}$, a contradiction. 
  \begin{center}
  \includegraphics[scale=0.67]{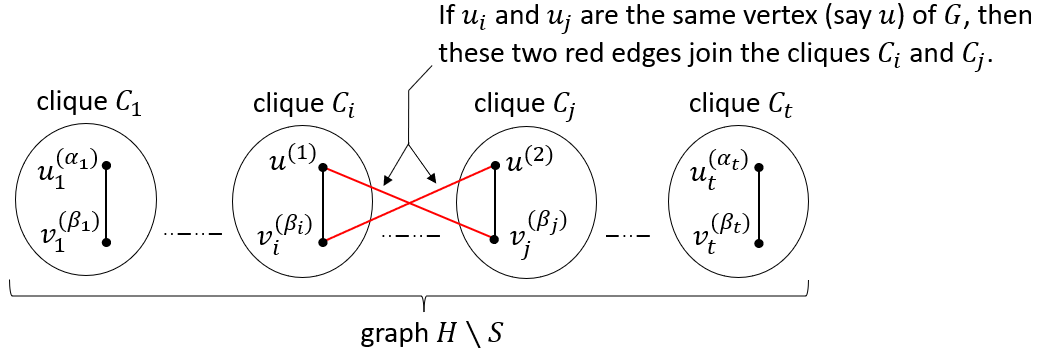}
  \end{center}
  Therefore, $u_1,\ldots, u_t$ are $t\geq z$ distinct vertices that form an independent set in $G$.     
  
  Now note that for any fixed $d$, by modifiying the reduction described above so that we have $d$ copies of each vertex, the reduced instance becomes triangle-free and $3d$-regular. It is straightforward to verify that the reduction remains valid with this modification, and therefore, we conclude that the $2$-\textsc{Eigenvalue Vertex Deletion} problem is \textsf{NP}-complete even on triangle-free and $3d$-regular graphs for any $d\geq 2$.


\textbf{Second reduction.}\\  Now we turn to the second claim in the Theorem. Let us describe a polynomial-time many-one reduction from \textsc{Independent Set} on planar cubic triangle-free graphs $\big($known to be \textsf{NP}-complete~\citep{IndesetNPhardness}$\big)$ to $2$-\textsc{Eigenvalue Vertex Deletion}. Consider an instance, say $(G,z)$, of \textsc{Independent Set}, where $G$ is a planar cubic triangle-free graph, say on $n$ vertices. Construct a graph, say $H$, as follows: For each vertex $v\in V(G)$, introduce six copies of $v$, say $v^{(1)},\ldots, v^{(6)}$, and make them pairwise adjacent to each other. Also, for each edge $e\in E(G)$, say with endpoints $u$ and $v$, make the $i^{th}$ copy of $u$, i.e., $u^{(i)}$, adjacent to the $i^{th}$ copy of $v$, i.e., $v^{(i)}$, for all $1\leq i\leq 6$. That is, $V(H) = \big\{v^{(i)}~|~v\in V(G)~\mbox{and}~1\leq i\leq 6\big\}$, and ${E(H) = \Big\{\{v^{(i)},v^{(j)}\}~|~v\in V(G)~\mbox{and}~1\leq i<j\leq 6\Big\}~ 
 \uplus  \Big\{\{u^{(i)},v^{(i)}\}~|~\{u,v\}\in E(G)~\mbox{and}~1\leq i\leq 6\Big\}}$. See \Cref{second reduction} for an example.
 
\begin{figure}
\centering  \includegraphics[scale=0.63]{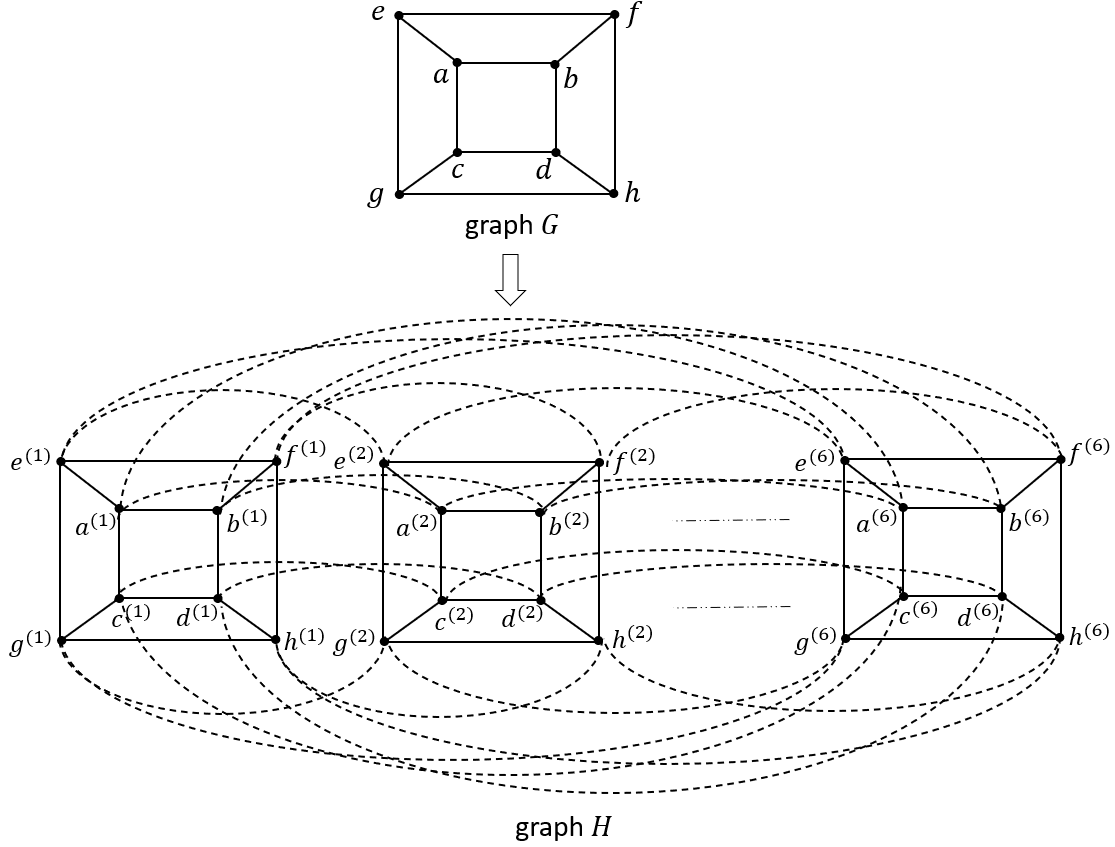}
\caption{An example illustrating the construction of $H$ from $G$ in the second reduction of \Cref{2EVD hard}.}
\label{second reduction}
\end{figure}

Note that for each vertex $v\in V(G)$ and each $1\leq i\leq 6$, the $i^{th}$ copy of $v$, i.e., $v^{(i)}$, has eight neighbours in $H$, namely the remaining five copies of $v$, and the $i^{th}$ copies of the three neighbours of $v$ in $G$. So, $H$ is an $8$-regular graph.
  
  Let us show that $G$ has an independent set of size $z$ if and only if $\big(H,6(n-z)\big)$ is a \textbf{\textsc{YES}} instance of $2$-\textsc{Eigenvalue Vertex Deletion}.
  
  \textbf{Forward direction}.\\
  Suppose that $G$ has an independent set, say $I$, of size $z$. Let $v_1,\ldots, v_z$ denote the vertices of $I$. Let $I'\subseteq V(H)$ denote the set that consists of all six copies of each of $v_1,\ldots, v_z$. That is,
  $$I' = \Big\{v_i^{(j)}~|~1\leq i\leq z~\mbox{and}~1\leq j\leq 6\Big\}$$
  As shown below, $H[I']$ is a disjoint union of $z$ cliques, each of size six, namely
  \begin{itemize}
  \item the clique on the six copies of $v_1$, i.e., $v_1^{(1)},\ldots, v_1^{(6)}$
  \item the clique on the six copies of $v_2$, i.e., $v_2^{(1)},\ldots, v_2^{(6)}$\\
  \vdots\\
  \vdots
  \item the clique on the six copies of $v_z$, i.e., $v_z^{(1)},\ldots, v_z^{(6)}$
  \end{itemize}
  
  \begin{center}
  \includegraphics[scale=0.6]{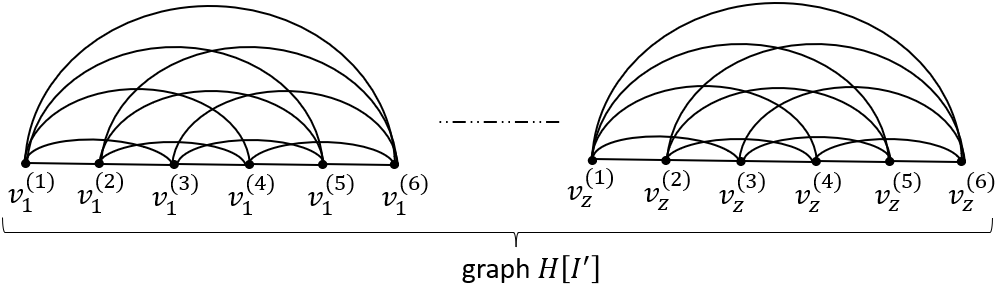}
  \end{center}

  The adjacency matrix of each of these $z$ cliques has two distinct eigenvalues, i.e., $-1$ and $5$. So, as $|V(H)\setminus I'|= 6(n-z)$, it follows that $\big(H,6(n-z)\big)$ is a \textbf{\textsc{YES}} instance of $2$-\textsc{Eigenvalue Vertex Deletion}. 
  
  \textbf{Reverse direction}.\\
  Suppose that $\big(H,6(n-z)\big)$ is a \textbf{\textsc{YES}} instance of $2$-\textsc{Eigenvalue Vertex Deletion}. That is, there exists $S\subseteq V(H)$ of size $\leq 6(n-z)$ such that the adjacency matrix of $H\setminus S$ has at most two distinct eigenvalues. Using Lemma \ref{2eval}, $H\setminus S$ is a disjoint union of equal-sized cliques, say $C_1,\ldots, C_t$, each of size $s$. We have $|V(H)\setminus S| = s\cdot t \geq 6z$. Let us show that $G$ has an independent set of size $z$. It is known that every planar cubic triangle-free graph on $n$ vertices has an independent set of size at least $\frac{3n}{8}$~\citep{LargeIndeSet}. So, if $z\leq \frac{3n}{8}$, we are done. Next, let us assume that $z>\frac{3n}{8}$.
  
  For any vertex $v\in V(G)$ and any $1\leq i<j\leq t$, it's not possible that each of $C_i$ and $C_j$ contains a copy of $v$; otherwise, there's an edge joining the cliques $C_i$ and $C_j$ in $H$. That is, each vertex of $G$ contributes its cop-y(ies) to at most one of the $t$ cliques $C_1,\ldots, C_t$. So, we get $t\leq n$. Now, we have $s\geq \big\lceil\frac{6z}{t}\big\rceil\geq \Big\lceil\frac{6\big(\frac{3n}{8}\big)}{n}\Big\rceil = 3$.
  
  For each $1\leq i\leq 6$, let $X_i$ denote the set of those vertices in $H\setminus S$ that are the $i^{th}$ copy of some vertex in $G$. We have $V(H)\setminus S = X_1\uplus\ldots \uplus X_{6}$. Now, as $|V(H)\setminus S|\geq 6z$, by pigeonhole principle, there exists $1\leq p\leq 6$ such that $X_{p}$ contains at least $z$ vertices, say $v_1^{(p)},\ldots v_z^{(p)}$. It suffices to show that $v_1,\ldots,v_z$ form an independent set in $G$.
  
  For the sake of contradiction, assume that there exist $1\leq i<j\leq z$ such that $v_i$ and $v_j$ are adjacent to each other in $G$. Then, $v_i^{(p)}$ and $v_j^{(p)}$ are adjacent to each other in $H$. So,  $v_i^{(p)}$ and $v_j^{(p)}$ belong to the same clique, say $C_{\ell}$, amongst the $t$ cliques $C_1,\ldots, C_t$. As $s\geq 3$, the clique $C_{\ell}$ has at least one vertex, say $u^{(q)}$, other than $v_i^{(p)}$ and $v_j^{(p)}$. These three vertices, i.e., $u^{(q)}$, $v_i^{(p)}$, $v_j^{(p)}$, are pairwise adjacent to each other in $H$. As shown below, this isn't possible.
  
  \begin{center}
  \includegraphics[scale=0.67]{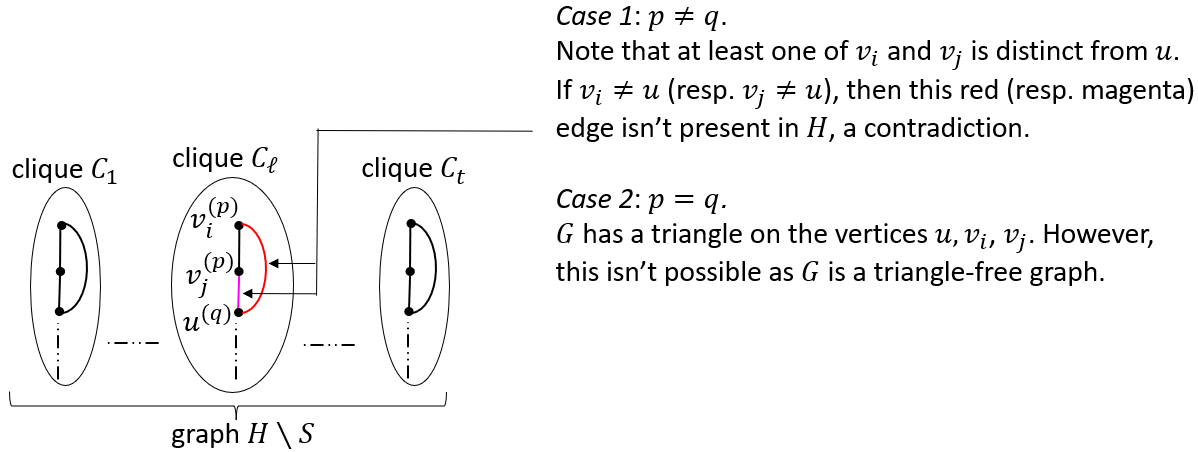}
  \end{center}
This concludes the proof of Theorem~\ref{2EVD hard}.
\end{proof}
  
Next, we note that $2$-EVD admits a branch-and-bound-based \textsf{FPT} algorithm that is similar in spirit to the naive branching algorithm for \textsc{Cluster Vertex Deletion}. As long as our instance has an induced path on three vertices $\{u,v,w\}$, we recursively solve the instances $(G \setminus \{u\},k-1)$, $(G \setminus \{v\},k-1)$ and $(G \setminus \{w\},k-1)$. Note that this branching algorithm enumerates all minimal subsets $S$ of size at most $k$ such that $G \setminus S$ is a disjoint union of cliques. At a leaf of any successful execution path of this branching algorithm, we are left with a subgraph $H$ of $G$ that is a cluster graph, and a (possibly reduced) budget $k^\prime \leq k$. At this point, we guess the value of $\ell$, and extend our solution greedily by: (a) deleting all cliques smaller than $\ell$, and (b) for any cliques of size, say $q$ where $q > \ell$, we delete an arbitrary subset of $(q-\ell)$ vertices. We have a valid solution at this if and only if there is some $\ell$ for which the cost of ``uniformizing'' the cluster graph $H$ to cliques of size $\ell$ is within the remaining budget $k^\prime$. 

\begin{restatable}{theorem}{twoEVDfpt}
\label{2EVD fpt}
$2$-\textsc{Eigenvalue Vertex Deletion} can be solved in time $\mathcal{O}^{\star}(3^k)$.
\end{restatable}
\begin{proof}
Let us describe a recursive branching algorithm. Consider an instance, say $(G, k)$, of $2$-\textsc{Eigenvalue vertex Deletion}. By Lemma \ref{2eval}, our goal is to decide whether we can delete at most $k$ vertices from $G$ to get a disjoint union of equal-sized cliques. First, we check if $G$ has an induced path on three vertices. This takes polynomial time.

  \noindent Case 1: $G$ has no induced path on three vertices\\
  The graph $G$ is a disjoint union of cliques, say $C_1, \ldots, C_t$, of sizes $s_1,\ldots, s_t$ respectively. We know that deleting the vertices of any solution results in a disjoint union of equal-sized cliques $\big($say, of size $x\big)$. Observe that for each $1\leq i\leq t$, 
  \begin{itemize}
  \item If $s_i\geq x$, then $s_i-x$ vertices of the clique $C_i$ are deleted, leaving behind $x$ of its vertices.
  \item If $s_i<x$, then the entire clique $C_i$, i.e., all its $s_i$ vertices, are deleted.
  \end{itemize}
  So, the overall solution size, i.e., total number of deleted vertices, is
  \[\sum_{\substack{1\leq i\leq t:\\ s_i\geq x}} (s_i-x) + \sum_{\substack{1\leq i\leq t:\\ s_i< x}} s_i =\sum_{i=1}^t s_i- x \cdot \mu(x)\]
  where $\mu(x)$ denotes the number of $s_i$'s amongst $s_1,\ldots,s_t$ such that $s_i\geq x$.
  
  Thus, the size of any minimum-sized solution is 
  \[\sum_{i=1}^t s_i - \max_{1\leq j\leq t} \big(s_j \cdot \mu(s_j)\big)\]
  
  If this size is $\leq k$, we return \textbf{\textsc{YES}}; otherwise, we return \textbf{\textsc{NO}}. This takes polynomial time.
  
  \noindent Case 2: $G$ has an induced path on three vertices, say $a-b-c$\\
  Note that any solution must pick at least one of its three vertices, i.e., $a, b, c$. So, if $k=0$, we return NO; otherwise, we guess a vertex that is picked into solution. That is, we branch as follows: In the first $\big($resp. second and third$\big)$ branch, we include the vertex $a$ (resp. $b$ and $c$) into solution, delete it from $G$, and reduce the parameter $k$ by $1$. It takes polynomial time to create the subproblems  $\big(G\setminus \{a\}, k-1\big)$, $\big(G\setminus\{b\}, k-1\big)$ and $\big(G\setminus \{c\}, k-1\big)$. Next, we run our algorithm on these three instances. If at least one of these three recursive calls returns \textbf{\textsc{YES}}, so do we; otherwise, we return \textbf{\textsc{NO}}.
  
  The depth of our search tree is at most $k$. Also, each of its internal nodes has three children. Therefore, it has at most $\mathcal{O}(3^k)$ nodes. Thus, as we spend polynomial time at each node, the overall running time is at most $\mathcal{O}^{\star}(3^k)$. This concludes the proof of Theorem~\ref{2EVD fpt}.
\end{proof}
We now show that $2$-\textsc{Eigenvalue Vertex Deletion} can be solved in polynomial time when the input is a forest. Let $(G,k)$ be an instance of $2$-EVD where $G$ is a forest. Note that if $S$ is a valid solution, then $G \setminus S$ is either independent or a disjoint collection of edges. 

Therefore, we can arrive at an optimal solution by computing the size of a maximum independent set and a maximum induced matching: this can be done in polynomial time on forests \citep{alt1991computing,zito2000linear}. We also note that a similar argument applies to $d$-regular graphs for $d \leqslant 2$.
  
\begin{restatable}{proposition}{twoEVDtrees}
  \label{2EVD trees}
  $2$-\textsc{Eigenvalue Vertex Deletion} admits polynomial time algorithms on forests and $d$-regular graphs for $d \leqslant 2$.
\end{restatable}
\begin{proof}
Consider an instance, say $(T,k)$, of $2$-\textsc{Eigenvalue Vertex Deletion}, where $T$ is a tree, say on $n$ vertices. Note that our goal is to decide if at least one of the following holds true:
  \begin{itemize}
  \item There exists $S\subseteq V(T)$ of size $\leq k$ such that $T\setminus S$ is a disjoint union of $1$-sized cliques. That is, $T$ has a vertex cover of size $\leq k$.
  \item There exists $S\subseteq V(T)$ of size $\leq k$ such that $T\setminus S$ is a disjoint union of $2$-sized cliques. That is, $T$ has an induced matching of size $\geq \frac{n-k}{2}$.
  \end{itemize}
  It is known that there is a polynomial time algorithm to find a minimum vertex cover  \cite{alt1991computing} and a maximum induced matching \cite{zito2000linear} when the input graph is a tree. So, $2$-\textsc{Eigenvalue Vertex Deletion} is polynomial-time solvable on trees. 
  
  Next, let us consider the case when input graph $G$ is $d$-regular with $d \leq 2$. Note that if $d=2$, then $G$ is disjoint union of cycles, and if $d=1$, then $G$ is disjoint union of edges. In the latter case, $G$ already has only two distinct eigenvalues i.e., $-1$ and $1$; thus, no vertex deletions are needed. Let us assume that $G$ is disjoint union of cycles, say $C_1,\ldots,C_t$. Note that

\begin{itemize}
    \item The minimum number of vertex deletions needed to get a disjoint union of three-sized cliques is \[k_1:=\underset{\substack{1\leq i\leq t: \\C_i \mbox{ is not a triangle}}}{\sum} \ell(C_i)\]

    \item The minimum number of vertex deletions needed to get a disjoint union of two-sized cliques is
    \[k_2:=\underset{\substack{1\leq i\leq t:\\\ell(C_i)\equiv 0~(mod~3)}}{\sum}\frac{\ell(C_i)}{3} +\underset{\substack{1\leq i\leq t:\\\ell(C_i)\equiv 1~(mod~3)}}{\sum}\frac{\ell(C_i)+2}{3}+\underset{\substack{1\leq i\leq t:\\\ell(C_i)\equiv 2~(mod~3)}}{\sum}\frac{\ell(C_i)+4}{3}\]
\end{itemize}
 where for every $1\leq i\leq t$, $\ell(C_i)$ denotes the length of the cycle $C_i$. 

 If $min(k_1, k_2)\leq k$, we return \textbf{\textsc{YES}}; otherwise, we return \textbf{\textsc{NO}}. Thus, $2$-\textsc{Eigenvalue Vertex Deletion} is polynomial-time solvable on $2$-regular graphs.
\end{proof}

\subsection{r-EVD for $r \geqslant 3$}

To demonstrate the hardness of $r$-EVD for any fixed $r \geqslant 3$, we give a reduction from \textsc{Vertex Cover on Cubic Graphs}. 

\begin{restatable}{theorem}{rEVDhard}
\label{rEVD hard} 
Let $r\geq 3$ be an integer. Then, $r$-\textsc{Eigenvalue Vertex Deletion} is \textsf{NP}-complete, even on bipartite graphs of maximum degree four.
\end{restatable}
\begin{proof}
Consider an instance, say $(G,k)$, of \textsc{Vertex Cover}, where $G$ is a cubic graph. Let us construct a graph, say $H$, as follows: For each vertex $v\in V(G)$, let us attach a path on $\ell:=\big\lfloor\frac{r-1}{2}\big\rfloor$ vertices $\big($denoted by \FiveStarOpen$_{v}$'s$\big)$ to $v$.
  
  \begin{center}
  \includegraphics[scale=0.67]{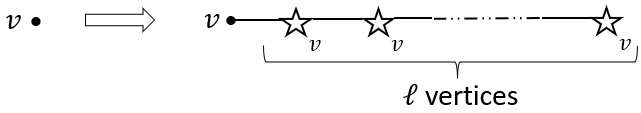}
  \end{center}
  
  Also, for each edge $e\in E(G)$, say with endpoints $u$ and $v$, let us subdivide $e$ into two edges using a vertex $\big($denoted by $\triangle_e$$\big)$, and then attach a path on $\ell-1$ vertices $\big($denoted by $\square_{e}$'s$\big)$ to $\triangle_e$, as shown below.
  \begin{center}
  \includegraphics[scale=0.7]{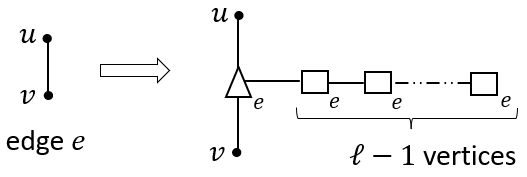}
  \end{center}
  
  See \Cref{rEVD construction} for an example.
 
  Observe that every cycle in $H$ is obtained from some cycle in $G$ by subdividing each of its edges into two edges. So, all cycles in $H$ have even length. That is, $H$ is a bipartite graph. Let us show that $G$ has a vertex cover of size $\leq k$ if and only if $(H,k)$ is a \textbf{\textsc{YES}} instance of $r$-\textsc{Eigenvalue Vertex Deletion}. 
  
  \noindent ($\Rightarrow$) Suppose that $G$ has a vertex cover, say $S$, of size $\leq k$. Observe that the graph $H\setminus S$ consists of
  three types of components (see Figure \ref{component types}).
  \begin{figure}[H]
 \centering \includegraphics[scale=0.67]{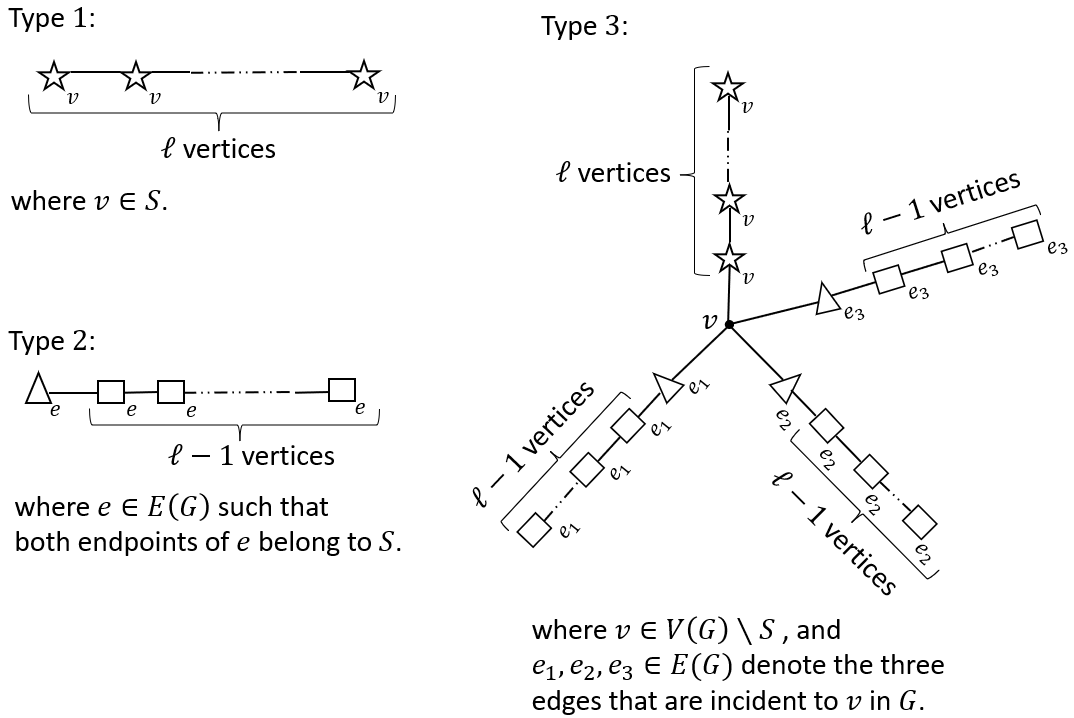}
\caption{Three types of components of $H\setminus S$ in the forward direction of the proof of Theorem~\ref{rEVD hard}.} \label{component types}
  \end{figure}
  
Each Type 1 or Type 2 component of $H\setminus S$ is a path on $\ell$ vertices; the $\ell$ eigenvalues of its adjacency matrix are $\lambda_1,\ldots, \lambda_{\ell}$, where $\lambda_j:= 2 \cos\big(\frac{\pi j}{\ell+1}\big)$ for each $1\leq j\leq \ell$. Next, consider any Type 3 component, say $C$, of $H\setminus S$. Let $v$ denote the central vertex of $C$. Note that $C\setminus \{v\}$ is a disjoint union of four paths, each on $\ell$ vertices. So, the adjacency matrix of $C\setminus \{v\}$, i.e., $A_{C\setminus \{v\}}$, has eigenvalues $\lambda_1,\ldots,\lambda_{\ell}$, each with multiplicity $4$. Also, using Cauchy interlacing, i.e., Lemma \ref{Cauchy interlacing}, we know that the eigenvalues of $A_{C\setminus \{v\}}$ interlace the eigenvalues of the adjacency matrix of $C$, i.e., $A_C$. Therefore, each of $\lambda_1,\ldots, \lambda_{\ell}$ is an eigenvalue of $A_{C}$ with multiplicity at least $3$ (as shown in Figure~\ref{interlacing}).

  \begin{figure}[H]
  \centering
  \includegraphics[scale=0.6]{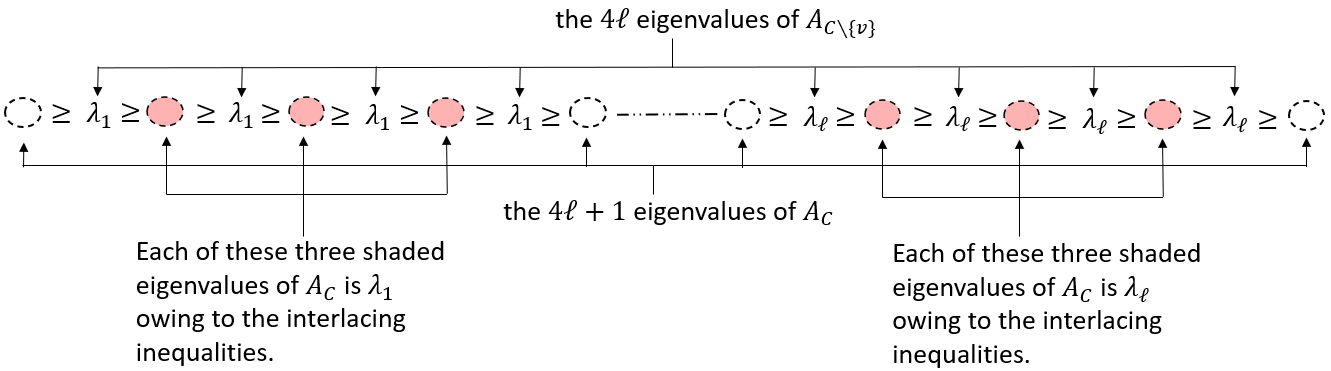}
 \caption{Eigenvalues of $A_{C\setminus \{v\}}$ and $A_C$ in the forward direction of the proof of Theorem~\ref{rEVD hard}.} \label{interlacing}
  \end{figure}
  Now, apart from $\lambda_1,\ldots, \lambda_{\ell}$, there are at most $(4\ell+1)-3\ell = \ell+1$ eigenvalues of $A_C$. So, the adjacency matrix of $H\setminus S$ has at most $\ell+(\ell+1) = 2\big\lfloor\frac{r-1}{2}\big\rfloor+1 \leq r$ distinct eigenvalues. Thus, $(H,k)$ is a \textbf{\textsc{YES}} instance of $r$-\textsc{Eigenvalue Vertex Deletion}.
  
  In \Cref{table-eigenvalues}, we provide the list of eigenvalues of the adjacency matrices of Type 1, 2, 3 components, when $\ell = 1, 2, 3, 4$.

  \begin{figure}[H]
  \centering
  \includegraphics[scale=0.67]{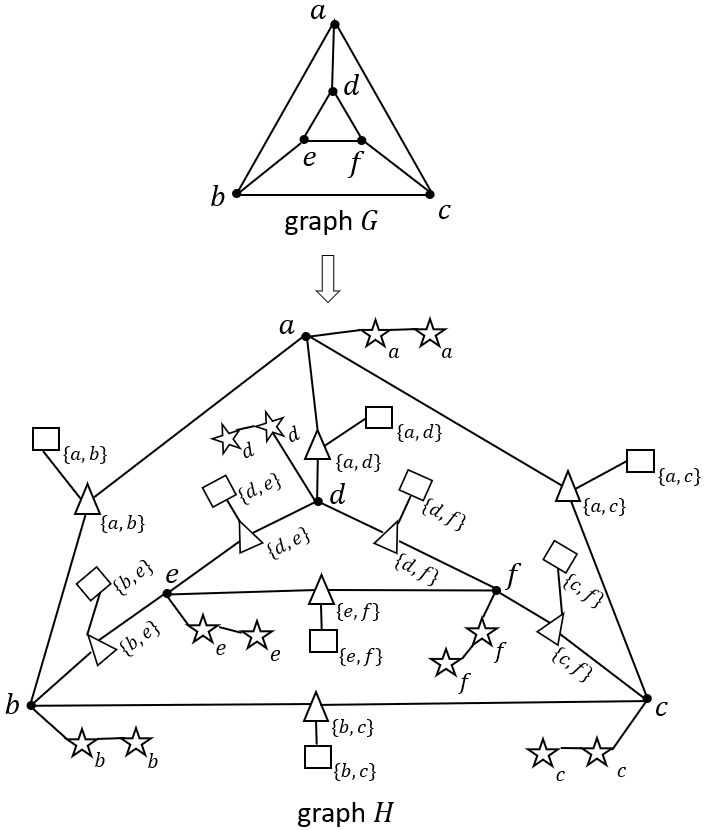}
  \caption{An example illustrating the construction of $H$ from $G$, when $\ell=2$, in \Cref{rEVD hard}.}
  \label{rEVD construction}
  \end{figure}
  \begin{sidewaysfigure}
  \centering
  \includegraphics[scale=0.65]{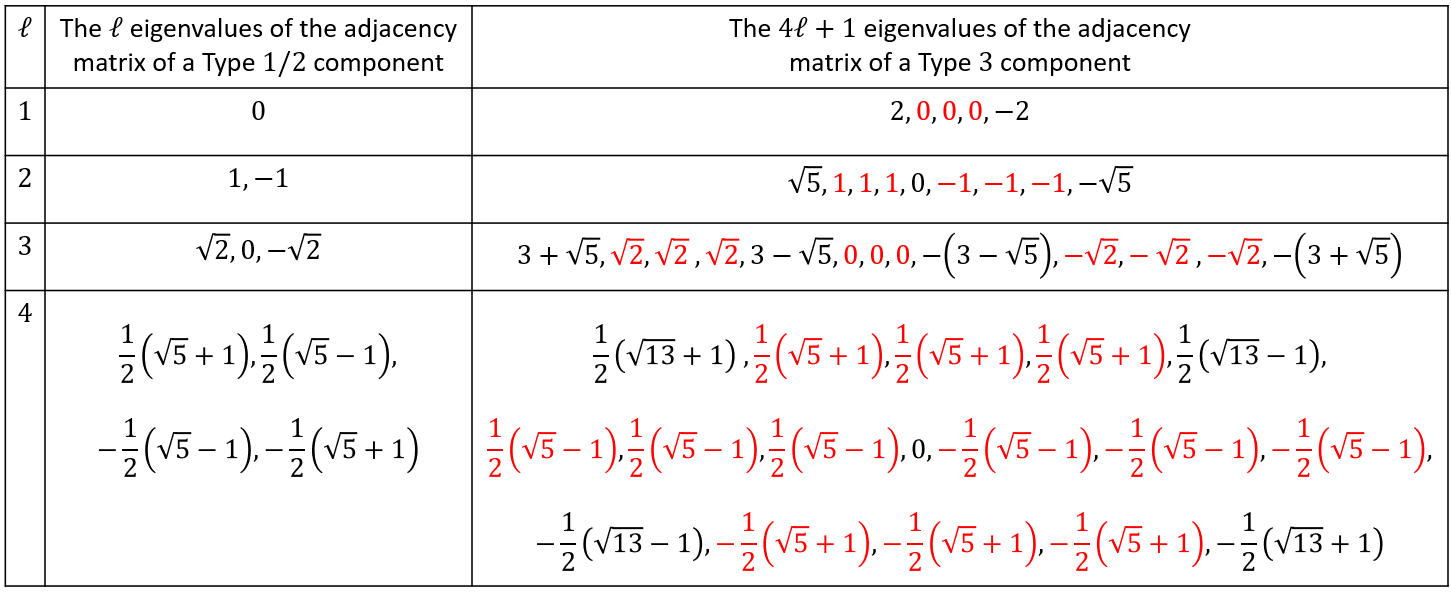}
  \caption{A table listing the eigenvalues of the adjacency matrices of Type 1, 2, 3 components, when $\ell = 1, 2, 3, 4$.}
  \label{table-eigenvalues}
  \end{sidewaysfigure}
  
  \noindent ($\Leftarrow$) Suppose that $(H,k)$ is a \textbf{\textsc{YES}} instance of $r$-\textsc{Eigenvalue Vertex Deletion}. That is, there exists $S\subseteq V(H)$ of size $\leq k$ such that the adjacency matrix of $H\setminus S$ has at most $r$ distinct eigenvalues. Let us construct a set $S'\subseteq V(G)$ as follows: For each vertex $v\in V(G)$ such that $v\in S$ or at least one of the $\ell$ \FiveStarOpen$_{v}$'s belongs to $S$, add $v$ to $S'$. Also, for each edge $e\in E(G)$ such that $\triangle_e$ $\in S$, arbitrarily pick an endpoint of $e$ and add it to $S'$. Note that $|S'|\leq k$. Now, it suffices to show that $S'$ is a vertex cover of $G$.
  
  For the sake of contradiction, assume that there exists an edge, say $e$, of $G$ such that neither of its endpoints, say $u$ and $v$, belongs to $S'$. Then, none of $u, v,$$\triangle_e$, the $\ell$ \FiveStarOpen$_{u}$'s and the $\ell$ \FiveStarOpen$_{v}$'s belong to $S$. Let $C$ denote the component of $H\setminus S$ that contains all these $2\ell+3$ vertices. Consider the following two vertices: i) the \FiveStarOpen$_u$ farthest from $u$, and ii) the \FiveStarOpen$_v$ farthest from $v$. Note that in $C$, the shortest path joining these two vertices has $2\ell+2$ edges.
  
  \begin{center}
  \includegraphics[scale=0.67]{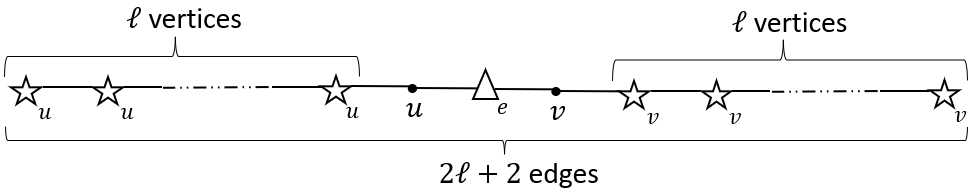}
  \end{center}
  
  So, the diameter of $C$ is at least $2\ell+2$. Thus, using Lemma \ref{diameter bound}, the adjacency matrix of $C$ (and so, $H\setminus S$) has at least $(2\ell+2)+1 = 2\big\lfloor\frac{r-1}{2}\big\rfloor+3>r$ distinct eigenvalues, a contradiction.
\end{proof}

Next, we show that $r$-EVD is \textsf{FPT} in the combined parmeter $k + \Delta(G)$, where $\Delta(G)$ is the maximum degree of $G$.

\begin{restatable}[]{theorem}{rEVDfpt}
\label{rEVD FPT}
Let $r \geq 3$ be an integer. Then, $r$-\textsc{Eigenvalue Vertex Deletion} admits an \textsf{FPT} algorithm running in time $\mathcal{O}^{\star}\Big((r+1)^{2k}\cdot 2^{k^2}\cdot \big(\Delta(G)\big)^{rk}\Big)$.
\end{restatable}

\begin{proof} (Sketch) Let $(G,k)$ be an instance of $r$-EVD. We claim that if $G$ has more than $(r+1) \cdot 2^k$ eigenvalues, then $G$ is a NO instance, and we can detect this upfront. The intuition is that the Cauchy interlacing structure (\Cref{Cauchy interlacing}) allows us to conclude that one vertex can reduce the number of distinct eigenvalues in the spectrum by a factor of at most half: so if there are ``too many'' distinct eigenvalues in the spectrum to begin with, $k$ deletions will not suffice to reduce the number of distinct eigenvalues substantially enough. We now quantify this argument: suppose, for the sake of contradiction, that $G$ has more than $(r+1) \cdot 2^k$ eigenvalues, and let $S \subseteq V(G)$ be a subset of at most $k$ vertices such that $A_{G \setminus S}$ has at most $r$ distinct eigenvalues. Denote the vertices of $S$ by $v_1, v_2, \ldots, v_t$, where $t \leqslant k$. By~\Cref{Cauchy interlacing} applied to $G$, $G \setminus \{v_1\}$, we know that the number of distinct eigenvalues in $G \setminus \{v_1\}$ is at least $\lfloor \frac{1}{2} \eta_G \rfloor$, where $\eta_G$ is  the number of distinct eigenvalues in $G$. Applying this argument iteratively to $G \setminus \{v_1\}$ and $G \setminus \{v_1,v_2\}$ and so on, it is clear that the number of distinct eigenvalues in $G \setminus S$ is at least $\frac{\eta_G}{2^k}-1$, but if $\eta_G > (r+1) \cdot 2^k$, then we have a contradiction. 

So we assume that $G$ has at most $(r+1) \cdot 2^k$ eigenvalues in its spectrum. Note that if $G$ has a shortest path $P$ with at least $r$ edges then any solution $S$ must contain one of the vertices of $P$ (c.f.~\Cref{diameter bound}). This gives us a branching strategy that can be executed in $\mathcal{O}^\star((r+1)^k)$ time. Let $(H,k^\prime)$ be an instance at a leaf of some successful execution path of this branching algorithm. Note that $H$ is a subgraph of $G$ whose diameter is at most $r-1$ and $k^\prime \leqslant k$ is a residual budget. 
  
Let $C$ be a connected component of $H$. Note that $|C| \leqslant \bigl(\Delta(G)\bigr)^r$, in other words, $H$ is a collection of ``small'' components. Note that if the spectrum of $H$ has more than $(r+1) \cdot 2^{k^\prime}$ eigenvalues, we say NO as before. On the other hand, if the spectrum of $H$ has at most $r$ eigenvalues, then we are already done. So the spectrum of $H$ has more than $r$ and at most $(r+1) \cdot 2^{k^\prime}$ eigenvalues. Otherwise, for the sake of analysis, assume that $(H,k^\prime)$ is a \textbf{\textsc{YES}}-instance with solution $S$. Note that there is an eigenvalue $\lambda$ that belongs to the spectrum of $H$ but not to the spectrum of $H \setminus S$. Note that there is at least one connected component $C$ such that $\lambda$ belongs to the spectrum of $H[C]$. Therefore, $S \cap C \neq \emptyset$. Our algorithm proceeds by guessing $\lambda$ and a choice of vertex from $S \cap C$, both of which we can afford because the spectrum of $H$ and the sizes of the components of $H$ are bounded by $(r+1) \cdot 2^{k^\prime}$ and  $|C| \leqslant \bigl(\Delta(G)\bigr)^r$ respectively. 
\end{proof}

\section{Reducing eigenvalues by adding edges}

We show that the $2$-\textsc{Eigenvalue Edge Addition} is \textsf{NP}-complete even on cluster graphs, and demonstrate a quadratic kernel in the standard parameter. Also, for any fixed $r\geq 3$, we show that the $r$-\textsc{EEA} problem is NP-complete. 

For the first result, we reduce from $3$-\textsc{Partition} which is known to be strongly \textsf{NP}-complete \cite[see][]{lewis1983michael}. The input for $3$-\textsc{Partition} consists of a set $T=\{ s_1,\ldots, s_{3n}\}$ and $b$, where $s_i$'s are positive integers from $\big(\frac{b}{4},\frac{b}{2}\big)$, $s_i$'s and $b$ are given in unary, and $\sum_{i=1}^{3n}s_i=nb$. The goal of this problem is to decide whether there $T$ can be partitioned into $n$ triplets such that the elements of any triplet sum up to $b$. The intuition for the reduction is the following: the reduced instance is a disjoint union of cliques whose sizes are $\{ s_1,\ldots, s_{3n}\}$ and a large number of cliques of size $b$. The idea is that a solution to the $3$-Partition instance can guide the smaller cliques into appropriate mergers so that all cliques have size $b$, and the ``large'' number of cliques of size $b$, combined with an appropriately chosen budget, essentially forces this solution structure in the reverse direction, allowing us to derive a solution for $3$-Partition. 

\begin{restatable}{theorem}{twoEEAhard}
\label{2EEA NPhard}
$2$-\textsc{Eigenvalue Edge Addition} is \textsf{NP}-complete, even when restricted to cluster graphs, forests, and $2$-regular graphs.
\end{restatable}

\begin{proof}
  We describe the hardness for cluster graphs. Consider an instance, say $(T, b)$, of $3$-\textsc{Partition}, where  $T=\{s_1,\ldots, s_{3n}\}$ such that i) $\frac{b}{4}< s_i<\frac{b}{2}$ for all $1\leq i\leq 3n$, and ii) $\sum_{i=1}^{3n}s_i = nb$.
  
  Let us construct a graph, say $G$, as follows: For every $1\leq i\leq 3n$, introduce a clique, say $C_i$, of size $s_i$. Also, add $M:= 3nb$  cliques, each of size $b$; let us refer to them as \emph{dummy cliques}. The graph $G$ is the disjoint union of $3n+M$ cliques, namely $C_1,\ldots, C_{3n}$ and the $M$ dummy cliques. Let us show that $(T,b)$ is a \textbf{\textsc{YES}} instance of $3$-\textsc{Partition} if and only if $(G,nb^2)$ is a \textbf{\textsc{YES}} instance of $2$-\textsc{Eigenvalue Edge Addition}.
  
  \noindent ($\Rightarrow$) Suppose that $(T,b)$ is a \textbf{\textsc{YES}} instance of $3$-\textsc{Partition}. Then, there exists a partition of $T$ into $n$ triplets, say $T = T_1\uplus\ldots\uplus T_n$, such that for every $1\leq i\leq n$, the elements of $T_i$ add up to $b$. That is, $s_{x_i} + s_{y_i} + s_{z_i} = b$, where $s_{x_i}, s_{y_i}, s_{z_i}$ denote the three elements of $T_i$.
  
  For every $1\leq i\leq n$, merge the three cliques $C_{x_i}, C_{y_i}, C_{z_i}$ into one clique, say $D_i$, as follows:
  \begin{itemize}
  \item Make every vertex of $C_{x_i}$ adjacent to every vertex of $C_{y_i}$.
  \item Make every vertex of $C_{x_i}$ adjacent to every vertex of $C_{z_i}$.
  \item Make every vertex of $C_{y_i}$ adjacent to every vertex of $C_{z_i}$. 
  \end{itemize}
  See \Cref{merging cliques} for an illustration.

  Note that the number of edges so added to $G$ is

  
  \[\sum_{i=1}^n\big(s_{x_i}\cdot s_{y_i}+s_{x_i}\cdot s_{z_i}+s_{y_i}\cdot s_{z_i}\big)<\sum_{i=1}^{n}\bigg(3\cdot \frac{b}{2}\cdot\frac{b}{2}\bigg)<nb^2\]
  
  For each $1\leq i\leq n$, the size of the clique $D_i$ is $s_{x_i}+s_{y_i}+s_{z_i} = b$. The resulting graph, say $H$, is the disjoint union of $n+M$ cliques, each of size $b$, namely  $D_1,\ldots, D_n$ and the $M$ dummy cliques. The adjacency matrix of $H$ has two distinct eigenvalues, i.e., $-1$ and $b-1$. Thus, $(G,nb^2)$ is a \textbf{\textsc{YES}} instance of $2$-\textsc{Eigenvalue Edge Addition}.

  \noindent ($\Leftarrow$) Suppose that $(G,nb^2)$ is a \textbf{\textsc{YES}} instance of $2$-\textsc{Eigenvalue Edge Addition}. That is, there exists $S\subseteq {V(G) \choose 2} \setminus E(G)$ of size $\leq nb^2$ such that adding the edges of $S$ to $G$ results in a graph, say $H$, whose adjacency matrix has at most two distinct eigenvalues. Using Lemma \ref{2eval}, the graph $H$ is a disjoint union of equal-sized cliques. Observe that each clique of $H$ is formed by merging some of the $3n+M$ cliques of $G$, namely $C_1,\ldots, C_{3n}$ and the $M$ $b$-sized dummy cliques. 
  
  First, let us show that no dummy clique  participates in a merger. That is, in $H$, each of the $M$ $b$-sized dummy cliques of $G$ remains as it is. For the sake of contradiction, assume that there exists a dummy clique that merges with some other clique(s) of $G$ to form a bigger (i.e., of size $>b$) clique of $H$. Then, as all cliques of $H$ have the same size, none of the other $M-1~~$ $b$-sized dummy cliques of $G$ can remain as it is. Now, as each of the $M$ $b$-sized dummy cliques participates in some merger, it is incident to $\geq b$ edges of $S$.  Also, every edge of $S$ is incident to at most two dummy cliques. Therefore, we get $|S|\geq \frac{Mb}{2}>nb^2$, a contradiction.

\begin{figure}[H]
  \centering
  \includegraphics[scale=0.66]{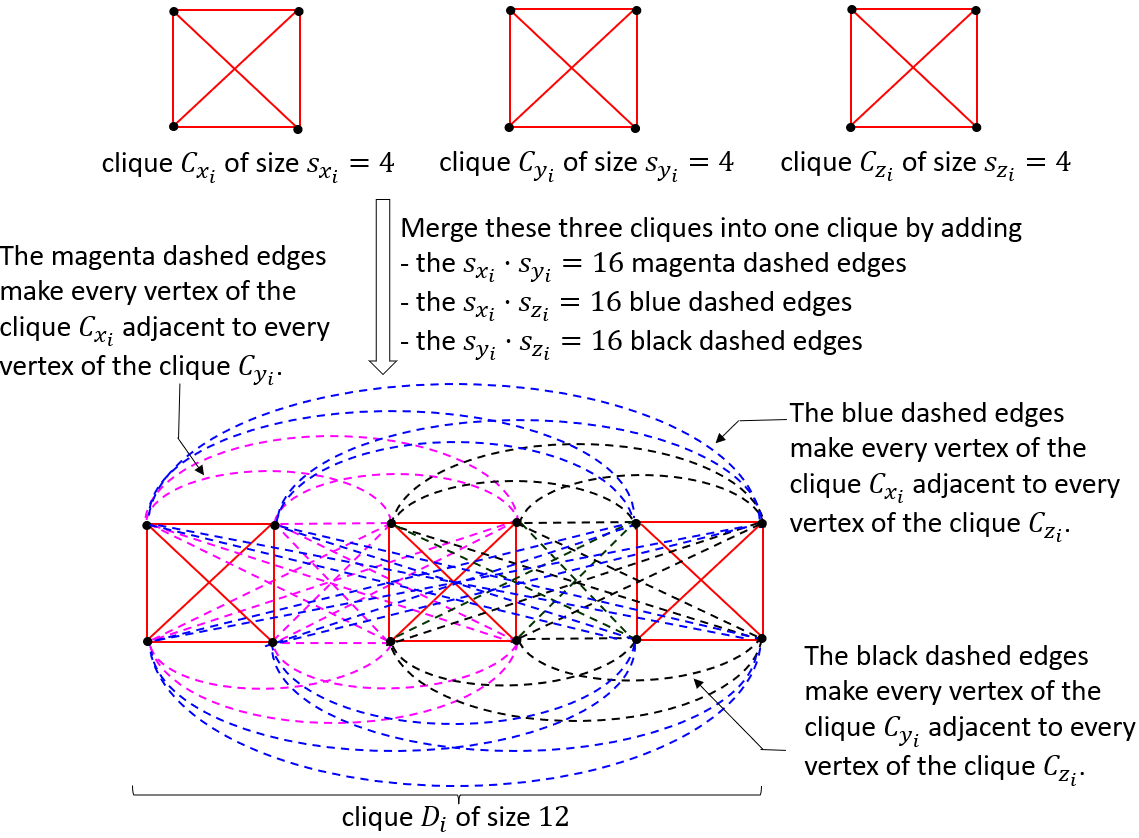}
  \caption{An illustration of the merger of the three cliques $C_{x_i}, C_{y_i}, C_{z_i}$, when $s_{x_i}=s_{y_i}=s_{z_i}=4$, in the forward direction in the proof of \Cref{2EEA NPhard}.}
  \label{merging cliques}
  \end{figure}
  
  Let $D_1, \ldots, D_t$ denote the equal-sized cliques of $H$ other than the $M$ dummy cliques. Note that their common size is the same as that of a dummy clique, i.e., $b$. Consider any $1\leq i\leq t$. The clique $D_i$ is formed by merging some (say $p_i$) of the $3n$ cliques $C_1,\ldots, C_{3n}$. Each of these $p_i$ cliques  has size $>\frac{b}{4}$ and $<\frac{b}{2}$. Also, their sizes add up to the size of the clique $D_i$, i.e., $b$. Therefore, we have $p_i\cdot \frac{b}{4} < b < p_i\cdot \frac{b}{2}$. So, we get $p_i=3$. Hence, each of the $t$ cliques $D_1,\ldots, D_t$ is obtained by merging three of the $3n$ cliques $C_1,\ldots, C_{3n}$. We have $t=n$. 
  
  Consider any $1\leq i\leq n$. Let $C_{x_i}, C_{y_i}, C_{z_i}$ denote the three cliques amongst $C_1,\ldots, C_{3n}$ whose merger forms the clique $D_i$. Let $T_i$ denote the triplet that consists of $s_{x_i}, s_{y_i}, s_{z_i}$. As the sizes of the cliques $C_{x_i}, C_{y_i}, C_{z_i}$ add up to the size of the clique $D_i$, we have $s_{x_i}+s_{y_i}+s_{z_i} = b$. That is, the elements of the triplet $T_i$ add up to $b$. Thus, as $T=T_1\uplus\ldots\uplus T_n$, it follows that $(T,b)$ is a \textbf{\textsc{YES}} instance of $3$-\textsc{Partition}. 
  
  In the reduction above, instead of adding a clique on $s_i$ vertices, we could instead add a cycle (resp. path) on $s_i$ vertices, and adjust the budget to account for the missing edges, thereby showing NP-completeness on $2$-regular graphs (resp. forests) as well.

\end{proof}

Our next result gives a quadratic kernel for $2$-\textsc{Eigenvalue Edge Addition}. Let $(G,k)$ be an instance of $2$-EEA. We only describe the main intuition of the kernel informally and defer a detailed argument to the appendix. Since we are only allowed to add edges, we ``might as well'' complete all the connected components of $G$ to cliques and adjust the budget accordingly. Thus, without loss of generality, $G$ is already a cluster graph. Some trivial cases are easily handled, such as: when we cannot afford to complete the original components of $G$ to cliques, or when we have no budget but cliques of different sizes, or when all cliques are already of the same size. 

Now, we are left with a situation where we have a non-trivial budget and cliques of at least two distinct sizes. Let the largest sized clique have $q$ vertices, and suppose we have $t$ cliques of size $p$ in $G$, denoted by $C_1, \ldots, C_t$, where $p < q$. Note that each of these cliques is merged into a larger clique after edges from any valid solution are added to $G$. In particular, if $S$ is a valid solution, at least $\frac{t \cdot p}{2}$ edges of $S$ are incident to vertices of $C_1 \cup \ldots \cup C_t$. Therefore, if $tp/2 > k$, we can say \textbf{\textsc{NO}}. This bounds the sizes of cliques with fewer than $q$ vertices. 

For the largest-sized cliques, note that if we have ``too many'' of them, then none of them are merged into a larger clique after edges from any valid solution are added to $G$. In particular, it can be shown that if there are $s$ cliques of size $q$, then if $sq > 2k$, then these cliques are untouched by any valid edge addition set of size at most $k$. This allows us to throw away most of them, preserving just enough to remember that the cliques must indeed remain untouched in any valid solution. This bounds the number of vertices among the largest sized clique.

Combining these arguments, the overall bound on the total number of vertices in the reduced instance turns out to be quadratic in $k$. 
\begin{restatable}{theorem}{twoEEAkernel}
\label{2EEA kernel}
$2$-\textsc{Eigenvalue Edge Addition} admits a kernel with $\mathcal{O}(k^2)$ vertices.
\end{restatable}
\begin{proof}
Consider an instance, say $(G,k)$, of $2$-\textsc{Eigenvalue Edge Addition}. Owing to Lemma \ref{2eval}, our goal is to decide if we can add $\leq k$ edges to $G$ to get a disjoint union of equal-sized cliques. Let us apply the following reduction rules $\big($in the specified order$\big)$: 
  
  \noindent \emph{Reduction rule 1:} Suppose that there's a component, say $C$, of $G$, that is not a clique. Then, add the missing ${|V(C)| \choose 2}-|E(C)|$ edges to turn $C$ into a clique, and reduce the parameter $k$ by ${|V(C)| \choose 2}-|E(C)|$.
  
  After exhaustively applying Reduction rule 1, $G$ is a disjoint union of cliques; say, it consists of $n_1$ cliques of size $x_1$, $n_2$ cliques of size $x_2$, $\ldots \ldots$, $n_t$ cliques of size $x_t$, where $x_1<x_2<\ldots\ldots <x_t$.
  
  \noindent \emph{Reduction rule 2:}
  \begin{itemize}
  \item If $k<0$, then return \textbf{\textsc{NO}}.
  \item If $k\geq 0$ and $t=1$, then return \textbf{\textsc{YES}}.
  \item If $k=0$ and $t\geq 2$, then return \textbf{\textsc{NO}}.
  \end{itemize}
  After applying Reduction rule 2, we have $k\geq 1$ and $t\geq 2$.
  
  \noindent \emph{Reduction rule 3:} If there exists an $1\leq i\leq t-1$ such that $n_i\cdot x_i>2k$, then return \textbf{\textsc{NO}}.
  
  \noindent \emph{Safeness of Reduction rule 3:}\\
  Suppose that $(G,k)$ is a \textbf{\textsc{YES}} instance. Then, there exists $S\subseteq {V(G) \choose 2} \setminus E(G)$ of size $\leq k$ such that adding the edges of $S$ to $G$ results in a disjoint union of equal-sized (say, of size $x$) cliques. Observe that each of these $x$-sized cliques is obtained by merging some cliques of $G$. Note that $x$ is at least the size of a largest clique in $G$. That is, we have $x\geq x_t$. Also, each of the smaller cliques of $G$, i.e., those of sizes $x_1, \ldots, x_{t-1}$, must participate in some merger.
  
  Now, consider any $1\leq i\leq t-1$. Each of the $n_i$ cliques of size $x_i$ is incident to $\geq x_i$ edges of $S$, for it must participate in some merger. Also, any edge of $S$ is incident to at most two of these $n_i$ cliques. Therefore, $|S|\geq \frac{n_i\cdot x_i}{2}$. So, as $|S|\leq k$, we get $n_i\cdot x_i\leq 2k$. Thus, Reduction rule 3 is safe.
  
  After applying Reduction rule 3, we have $n_i\cdot x_i\leq 2k$ for all $1\leq i\leq t-1$. Also, as $x_1,\ldots,x_{t-1}$ are $t-1$ distinct integers in the interval $[1,2k]$, we get $t-1\leq 2k$.
  
  \noindent \emph{Reduction rule 4:} Suppose that $n_t\cdot x_t>2k$. Then, remove all but $\frac{2k+1}{x_t}$ cliques of size $x_t$ from $G$.
  
  \noindent \emph{Safeness of Reduction rule 4:}\\
  If $n_t\cdot x_t>2k$, then in any solution, none of the $n_t$ cliques of size $x_t$ participate in a merger. That is, each of them remains as is after the edge additions, and each merger $\big($involving the remaining cliques, i.e., those of sizes  $x_1, \ldots, x_{t-1}\big)$ results in an $x_t$-sized clique. This is because if any clique of size $x_t$ gets to participate in a merger, then each of the remaining $n_t-1$ cliques of size $x_t$ must also participate in some merger $\big($because all cliques have the same size after the edge additions$\big)$, thereby needing  $\geq \frac{n_t\cdot x_t}{2}>k$ edge additions. 
  
  Also, we have $n_t\cdot x_t>2k$ before, as well as after, applying Reduction rule 4. Therefore, it follows that any solution before applying RR4 remains a solution after applying Reduction rule 4, and vice versa. Thus, Reduction rule 4 is safe.
  
  If Reduction rule 4 wasn't invoked, then $n_t\cdot x_t \leq 2k$; otherwise, after applying Reduction rule 4, we get $n_t\cdot x_t = 2k+1$. 
  
  Finally, the number of vertices in $G$ is at most \[n_1\cdot x_1+\ldots\ldots + n_{t-1}\cdot x_{t-1} + n_t\cdot x_t\leq (t-1)\cdot 2k +  (2k+1)\leq 4k^2+2k+1.\]
  
  This concludes the proof of Theorem \ref{2EEA kernel}. 
\end{proof}
Next, we show that $r$-\textsc{EEA} is NP-complete for every fixed $r\geq 3$.

\begin{restatable}{theorem}{rEEAhard}
\label{rEEA NP hardness}
Let $r\geq 3$ be an integer. Then, $r$-\textsc{Eigenvalue Edge Addition} is NP-complete.
\end{restatable}
\begin{proof}
Consider an instance, say $(T,b)$, of $3$-\textsc{Partition}, where $T=\{s_1,\ldots, s_{3n}\}$. Construct a graph, say $G$, as follows: For every $1\leq i\leq 3n$, introduce a clique, say $C_i$, of size $s_i$. Add $2nb^2+1$ cliques, each of size $b$. Also, for every $0\leq i\leq r-3$, add $2nb^2+1$ cliques, each of size $L+i$, where $L:=6nb^3$. Let us show that $(T,b)$ is a YES instance of $3$-\textsc{Partition} if and only if $(G,nb^2)$ is a YES instance of $r$-\textsc{Eigenvalue Edge Addition}.

($\Rightarrow$) Suppose that $(T,b)$ is a \textbf{\textsc{YES}} instance of $3$-\textsc{Partition}. Then, as described in the proof of Theorem~\ref{2EEA NPhard}, we add $\leq nb^2$ edges to merge the cliques $C_1,\ldots,C_{3n}$ using the $3$-\textsc{Partition} solution. The resulting graph, say $H$, is the disjoint union of i) $n+2nb^2+1$ cliques of size $b$ (each contributing eigenvalues $-1$ and $b-1$), and ii) $2nb^2+1$ cliques of size $L+i$ (each contributing eigenvalues $-1$ and $L+i-1$) for each $0\leq i\leq r-3$. So, the adjacency matrix of $H$ has $r$ distinct eigenvalues, namely $-1,b-1, L-1, L, L+1,\ldots, L+(r-4)$. Thus, $(G,nb^2)$ is a \textbf{\textsc{YES}} instance of $r$-\textsc{Eigenvalue Edge Addition}.

($\Leftarrow$) Suppose that $(G,nb^2)$ is a \textbf{\textsc{YES}} instance of $r$-\textsc{Eigenvalue Edge Addition}. That is, there exists $S\subseteq \binom{V(G)}{2}\setminus E(G)$ of size $\leq nb^2$ such that adding the edges of $S$ to $G$ results in a graph, say $H$, whose adjacency matrix has at most $r$ distinct eigenvalues. Note that any edge of $S$ is incident to at most two of the $2nb^2+1$ cliques of size $b$. So, as $|S|\leq nb^2$, there exists a $b$-sized clique that is not incident to any edge of $S$; this clique survives as a component in $H$, contributing eigenvalues $-1$ and $b-1$. Similarly, for each $0\leq i\leq r-3$, at least one clique of size $L+i$ survives as a component in $H$, contributing eigenvalues $-1$ and $L+i-1$. So, the $r$ distinct eigenvalues of the adjacency matrix of $H$ are $-1, b-1, L-1, L, L+1, \ldots, L+(r-4)$. Now, using Lemma~\ref{smallest ev -1}, it is clear that the graph $H$ must be a disjoint union of some cliques, whose sizes are $b, L, L+1, \ldots, L+(r-3)$.

Note that any clique of size $\geq L$ would need at least $L>nb^2$ edges to participate in a merger. So, as $|S|\leq nb^2$, all cliques of sizes $L, L+1, \ldots, L+(r-3)$ must remain intact in $H$. So, the remaining cliques (i.e., $C_1,\ldots C_{3n}$ and the $2nb^2+1$ cliques of size $b$) of $G$ must merge to give some cliques of sizes $b, L, L+1, \ldots, L+(r-3)$. However, note that all these cliques together contain $s_1+\ldots+s_{3n}+b\cdot(2nb^2+1)=nb+ b\cdot (2nb^2+1)<L$ vertices. Thus, their merger would only give $b$-sized cliques; that is, no clique of size $L, L+1, \ldots,$ or $L+(r-3)$ is produced by such mergers. Now, as described in the proof of Theorem~\ref{2EEA NPhard}, each such $b$-sized clique must be obtained by merging exactly three cliques amongst $C_1,\ldots, C_{3n}$, thereby showing that $(T,b)$ is a \textbf{\textsc{YES}} instance of $3$-\textsc{Partition}. 
\end{proof}

\section{Reducing eigenvalues by deleting edges}

In this section, we consider the $r$-\textsc{Eigenvalue Edge Deletion} problem. We defer the \textsf{NP}-completeness of $2$-EED to the proof of \Cref{2EEE NP hard}, where the hardness is implicit. In this section, we present an $\mathcal{O}^{*}(2^k)$-time \textsf{FPT} algorithm for $2$-EED and show that it can be solved in polynomial time on triangle-free graphs. Finally, we prove that $r$-EED is \textsf{NP}-complete for any fixed $r \geq 3$.

The \textsf{FPT} algorithm is similar in spirit to the one we use in the proof of~\Cref{2EVD fpt}: we branch on induced paths of length three, except we now have a choice of two edges instead of three vertices. In particular, if $P$ is an induced path on $\{a,b,c\}$ with edges $\{a,b\}$ and $\{b,c\}$, we recursively solve the instances $(G\setminus \{a,b\},k-1)$ and $(G \setminus \{b,c\},k-1)$. 

At the leaves of successful execution paths of this branching algorithm, as before, we have cluster graphs where the cliques are not necessarily of the same size, and a residual budget. Let $(H,k^\prime)$ denote such an instance, where $H$ is a subgraph of $G$ consisting of $t$ cliques of sizes $s_1, \ldots, s_t$, and $k^\prime \leqslant k$ is the residual budget. Note that if $S$ is such that $G \setminus S$ is a collection of $x$-sized cliques for some $x$, then $x$ must divide each $s_i$. We show that for an optimal choice of $S$, $x$ is the GCD of the $s_i$'s. Based on this, it is straightforward to check if the residual budget is sufficient or not. 


\begin{restatable}{theorem}{twoEEDfpt}
\label{2EED fpt}
$2$-\textsc{Eigenvalue Edge Deletion} admits an algorithm with running time $\mathcal{O}^{*}(2^k)$.
\end{restatable}
\begin{proof}
Let us describe a recursive branching algorithm. Consider an instance, say $(G, k)$, of $2$-\textsc{Eigenvalue Edge Deletion}. Owing to Lemma \ref{2eval}, our goal is to decide whether we can delete at most $k$ edges from $G$ to get a disjoint union of equal-sized cliques. First, we check if $G$ has an induced path on three vertices. This takes polynomial time. 
  
  \noindent \emph{Case 1: $G$ has no induced path on three vertices:}\\
  The graph $G$ is a disjoint union of cliques, say $C_1, \ldots, C_t$, of sizes $s_1,\ldots, s_t$ respectively. Observe that deleting the edges of any solution breaks each of these $t$ cliques into equal-sized cliques $\big($say, of size $x\big)$. That is, for every $1\leq i\leq t$, it breaks the clique $C_i$ into $\frac{s_i}{x}$ cliques, each of size $x$. As each of these $\frac{s_i}{x}$ cliques has ${x \choose 2}$ edges, the number of edges deleted from the clique $C_i$ is $${s_i \choose 2}- \frac{s_i}{x} {x \choose 2} = \frac{s_i\big(s_i-x\big)}{2}$$
  So, larger $x$ corresponds to smaller solutions, i.e., fewer edge deletions. Also, $x$ must divide each of $s_1, \ldots, s_t$. Therefore, for any minimum-sized solution, we have $x = gcd(s_1,\ldots, s_t)$, and its size is
  $$\sum_{i=1}^{t} \frac{s_i\big(s_i-gcd(s_1,\ldots, s_t)\big)}{2}$$
  If this size is at most $k$, we return \textbf{\textsc{YES}}; otherwise, we return \textbf{\textsc{NO}}. This takes polynomial time.
  
  See \Cref{breaking cliques} for an example.
  \begin{figure}[H]
   \thinspace
   \thinspace
   \thinspace
   \thinspace
   \thinspace
   \thinspace
   \thinspace
   \thinspace
   \thinspace
   \thinspace
   \thinspace
   \thinspace
   \thinspace
   \thinspace
   \thinspace \includegraphics[scale=0.66]{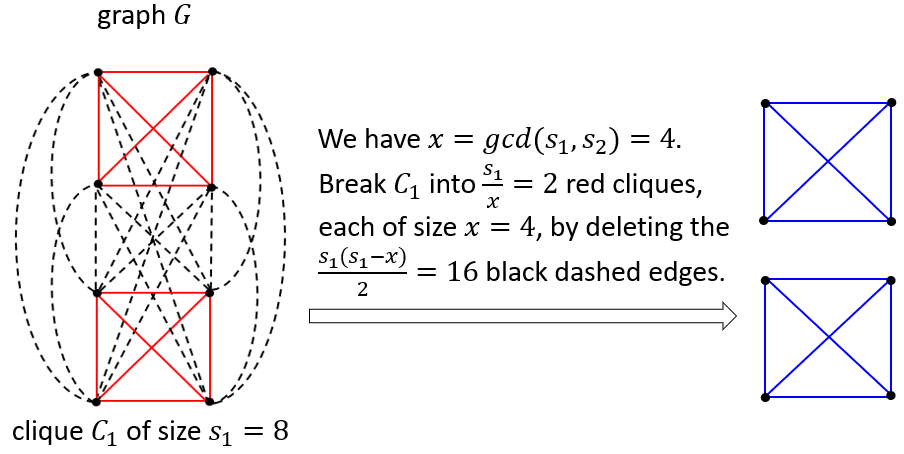}
  \includegraphics[scale=0.66]{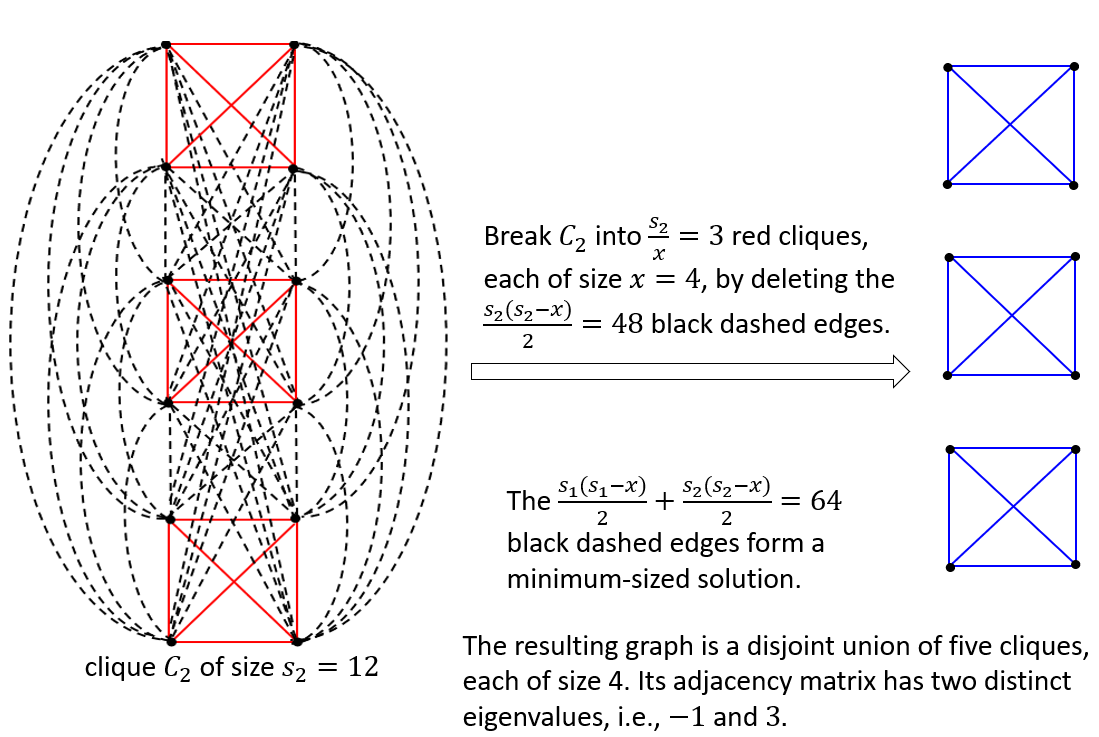}
  \caption{An example illustrating the breaking of cliques in \Cref{2EED fpt}.}
  \label{breaking cliques}
  \end{figure}

\noindent \emph{Case 2: $G$ has an induced path on three vertices, say $a-b-c$:}\\
Note that any solution must pick at least one of its two edges, i.e., $\{a, b\}$ and $\{b, c\}$. So, if $k=0$, we return NO; otherwise, we guess an edge that is picked into solution. That is, we branch as follows: In the first $\big($resp. second$\big)$ branch, we include the edge $\{a,b\}$ $\big($resp. $\{b,c\}\big)$ into solution, remove it from $G$, and reduce the parameter $k$ by $1$. It takes polynomial time to create the sub-problems $\big(G - \{a, b\}, k-1\big)$ and $\big(G - \{b, c\}, k-1\big)$. Next, we run our algorithm on these two instances. If at least one of these two recursive calls returns \textbf{\textsc{YES}}, so do we; otherwise, we return \textbf{\textsc{NO}}.
  
  The depth of our search tree is at most $k$. Also, each of its internal nodes has two children. Therefore, it has at most $\mathcal{O}(2^k)$ nodes. Thus, as we spend polynomial time at each node, the overall running time is at most $\mathcal{O}^{\star}(2^k)$.
  
  This concludes the proof of Theorem \ref{2EED fpt}.
\end{proof}

Our next claim takes advantage of the fact that the sizes of the cliques after the removal of any solution is at most two when the input graph is triangle-free and we are only allowed to delete edges. Therefore, the value of the optimal solution is $|E(G)| - |V(G)|/2$ if $G$ has a perfect matching and $|E(G)|$ otherwise. The result follows from the fact that the existence of a perfect matching can be determined in polynomial time \cite{maximumMatching}.

\begin{restatable}[]{proposition}{EEDtrianglefree}
  \label{EED trianglefree}
$2$-\textsc{Eigenvalue Edge Deletion} is polynomial time solvable on triangle-free graphs.
\end{restatable}

Now, we show that $r$-EED is \textsf{NP}-complete by reducing it from \textsc{Partition into Triangles} on graphs of clique number 3 which is known to be \textsf{NP}-complete \cite{custic2015geometric}. The input for \textsc{Partition into Triangles} is a graph $G$, and the goal is to decide whether $V(G)$ can be partitioned into $\frac{|V(G)|}{3}$ triplets such that every triplet induces a triangle in $G$.

\begin{restatable}{theorem}{reedhard}
\label{rEED NP hard} 
Let $r\geq 3$ be an integer. Then, $r$-\textsc{Eigenvalue Edge Deletion} is \textsf{NP}-complete. 
\end{restatable}

\begin{proof} 
Let us describe a polynomial-time many-one reduction from \textsc{Partition into Triangles} on graphs of clique number 3 to $r$-\textsc{Eigenvalue Edge Deletion}. Consider an instance, say $G$, of \textsc{Partition into Triangles}, where $G$ is a graph, say on $n$ vertices and $m$ edges, with clique number $3$. Let us construct a graph, say $H$, from $G$, as follows: First, we add $G$ as it is. Next, for each $3\leq i\leq r+1$, we introduce $M:=m-n+1$ cliques, each of size $i$; let us refer to these cliques as \emph{dummy cliques}. That is, the graph $H$ is the  disjoint union of the graph $G$, $M$ dummy cliques of size $3$, $M$ dummy cliques of size $4$, $\ldots \ldots$, $M$ dummy cliques of size $r+1$. We set the budget to be $m-n$. Let us show that $G$ has $\frac{n}{3}$ pairwise vertex disjoint triangles if and only if $(H,m-n)$ is a \textbf{\textsc{YES}} instance of $r$-\textsc{Eigenvalue Edge Deletion}.
  
  \noindent $(\Rightarrow)$ Suppose that $G$ has $\frac{n}{3}$ pairwise vertex disjoint triangles, say $T_1, \ldots, T_{n/3}$. Let $S$ denote the set that consists of those $m-n$ edges of $G$ that do not belong to any of these $\frac{n}{3}$ triangles. Note that the graph $H\setminus S$ is the disjoint union of
  \begin{itemize}
  \item $\frac{n}{3}+M$ triangles, namely $T_1,\ldots, T_{n/3}$ and the $M$ dummy cliques of size $3$. They contribute two distinct eigenvalues, i.e., $-1$ and $2$.
  \item $M$ dummy cliques of size $4$. They contribute two distinct eigenvalues, i.e., $-1$ and $3$.\\
  \vdots\\
  \vdots
  \item $M$ dummy cliques of size $r+1$. They contribute two distinct eigenvalues, i.e., $-1$ and $r$.
  \end{itemize}
  So, the adjacency matrix of the graph $H\setminus S$ has $r$ distinct eigenvalues, namely $-1, 2, 3, \ldots\ldots, r$. Thus, $(H,m-n)$ is a \textbf{\textsc{YES}} instance of $r$-\textsc{Eigenvalue Edge Deletion}. 
  
  \noindent $(\Leftarrow)$ Suppose that $(H,m-n)$ is a \textbf{\textsc{YES}} instance of $r$-\textsc{Eigenvalue Edge Deletion}. That is, there exists $S\subseteq E(H)$ of size $\leq m-n$ such that the adjacency matrix of the graph obtained by deleting the edges of $S$ from $H$ has $\leq r$ distinct eigenvalues.
  
  Consider any $3\leq i\leq r+1$. Note that the number of $i$-sized dummy cliques, i.e., $M$, is $>m-n\geq |S|$. So, there's at least one $i$-sized dummy clique, say $C_i$, such that none of its edges is deleted. That is, no edge of $C_i$ belongs to $S$ and thus, it appears as a component of the graph $H\setminus S$, thereby contributing two distinct eigenvalues, namely $-1$ and $i-1$. Thus, it follows that the adjacency matrix of the graph $H\setminus S$ must have $-1, 2, 3, \ldots, r$ as its $r$ distinct eigenvalues. 
  
  Now, using Lemma \ref{smallest ev -1}, it is clear that the graph $H\setminus S$ must be a disjoint union of some cliques, whose sizes are $3, 4, \ldots, r+1$. So, as $G$ has clique number $3$, after removing those edges of $G$ that belong to $S$, we're left with $\frac{n}{3}$ pairwise vertex-disjoint triangles of $G$, as desired.
  
  This concludes the proof of Theorem \ref{rEED NP hard}.
\end{proof}

\section{Reducing eigenvalues by editing edges}

In this section, we show that $2$-\textsc{Eigenvalue Edge Editing} is \textsf{NP}-complete. We give a reduction from \textsc{Partition into Triangles}.

\begin{restatable}[]{theorem}{twoeee}
\label{2EEE NP hard}
$2$-\textsc{Eigenvalue Edge Editing} is \textsf{NP}-complete.
\end{restatable}

\begin{proof}
  Let us describe a polynomial-time many-one reduction from \textsc{Partition into Triangles} to $2$-\textsc{Eigenvalue Edge Editing}. Consider an instance, say $G$, of \textsc{Partition into Triangles}, where $G$ is a graph on $n$ vertices and $m$ edges. Let us construct a graph $H$ based on $G$ as follows: for every vertex $v\in V(G)$, attach two triangles  to $v$, as shown below.
  \begin{center}
  \includegraphics[scale=0.65]{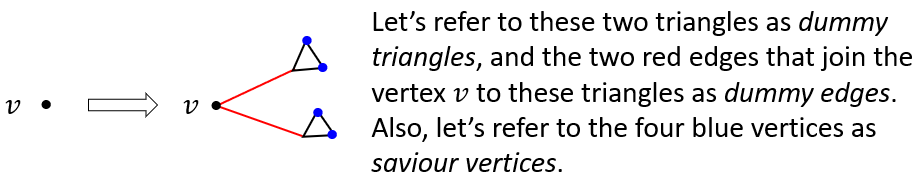}
  \end{center}

  See Figure \ref{2EEE example} for an illustration.

\begin{figure}
\centering
\includegraphics[scale=0.65]{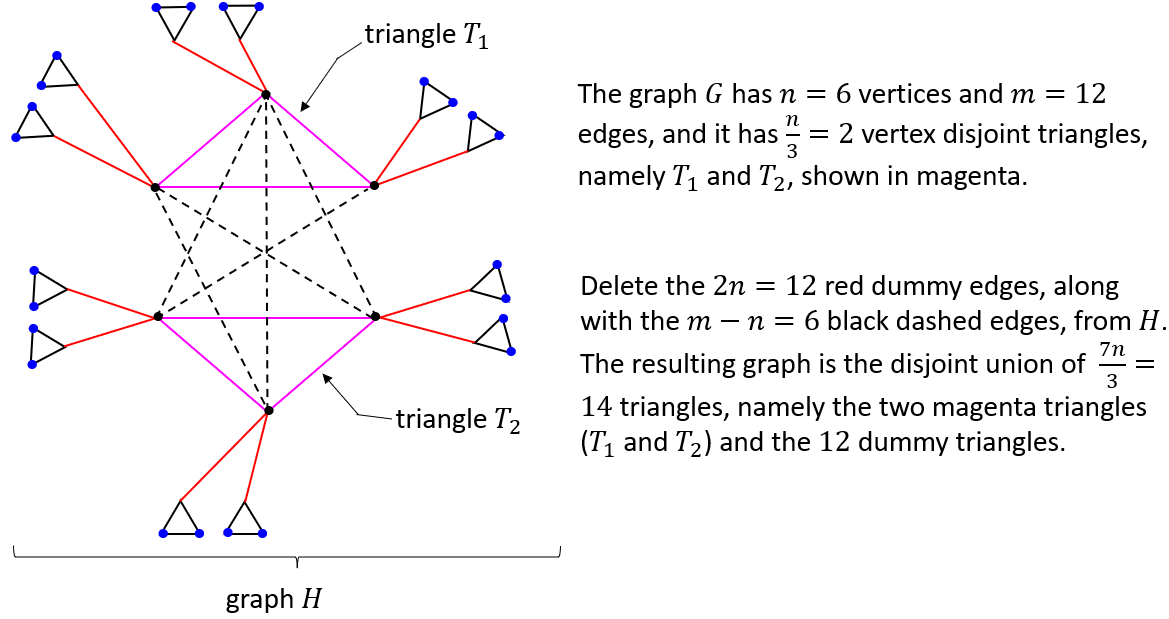}
\caption{An example illustrating the construction in \Cref{2EEE NP hard}.}
\label{2EEE example}
\end{figure}

  Note that $|V(H)|=7n$ and $|E(H)|=m+8n$. Let us show that $G$ has $\frac{n}{3}$ pairwise vertex disjoint triangles if and only if $(H,m+n)$ is a \textbf{\textsc{YES}} instance of $2$-\textsc{Eigenvalue Edge Editing}.
  
  \noindent ($\Rightarrow$) Suppose that $G$ has $\frac{n}{3}$ pairwise vertex disjoint triangles, say $T_1,\ldots, T_{n/3}$. Let $S\subseteq E(H)$ denote the set that  consists of the $2n$ dummy edges, along with those $m-n$ edges of $G$ that do not belong to any of these $\frac{n}{3}$ triangles. Note that the graph $H\setminus S$ is the disjoint union of $\frac{7n}{3}$ triangles, namely $T_1,\ldots, T_{n/3}$ and the $2n$ dummy triangles. Its adjacency matrix has two distinct eigenvalues, i.e., $-1$ and $2$. Thus, $(H,m+n)$ is a \textbf{\textsc{YES}} instance of $2$-\textsc{Eigenvalue Edge Editing}.
  

  \noindent $(\Leftarrow)$: Suppose that $(H,m+n)$ is a \textbf{\textsc{YES}} instance of $2$-\textsc{Eigenvalue Edge Editing}. That is, there exist $D\subseteq E(H)$ and $A\subseteq {V(H) \choose 2}\setminus E(H)$ such that: i) $|A|+|D|\leq m+n$, and ii) deleting the edges of $D$ from $H$, and adding the edges of $A$ to $H$, results in a graph, say $H'$, whose adjacency matrix has at most two distinct eigenvalues. Using Lemma \ref{2eval}, the graph $H'$ is a disjoint union of equal-sized cliques $\big($say, of size $x\big)$. As each of these $\frac{|V(H)|}{x}$ cliques has $ {x \choose 2}$ edges, the number of edges in $H'$ is
  $$\frac{|V(H)|}{x}\cdot {x \choose 2} = \frac{7n(x-1)}{2}$$ Also, we have $|E(H)|+|A|-|D|=|E(H')|$.
  Therefore, 
  \begin{equation}
  \label{eq1}
  (m+8n)+|A|-|D|=\frac{7n(x-1)}{2}
  \end{equation}
  Adding (\ref{eq1}) to the inequality $|A|+|D|\leq m+n$, we get 
  \begin{equation}
  \label{eq2}
  |A|\leq \frac{7n(x-3)}{4}
  \end{equation}
  Note that each saviour vertex has degrees $2$ and $x-1$ in $H$ and $H'$ respectively. So, each of the $4n$ saviour vertices is incident to $\geq x-3$ added edges $\big($i.e., edges of $A\big)$. Also, any edge of $A$ is incident to at most two saviour vertices. Therefore, 
  \begin{equation}
  \label{eq3}
  |A|\geq \frac{4n(x-3)}{2}
  \end{equation}
  Using (\ref{eq2}) and (\ref{eq3}), we get $x=3$ and $|A|=0$. Thus, the graph $H'$ is a disjoint union of $\frac{7n}{3}$ triangles, obtained from $H$ by only edge deletions: in other words, no edge additions are involved. This implies that we have $\frac{7n}{3}$ pairwise vertex disjoint triangles, say $T_1,\ldots, T_{7n/3}$, of the $7n$-vertex graph $H$. Note that the vertices of any dummy triangle belong to a unique triangle $\big($i.e., the dummy triangle itself$\big)$ in $H$. So, amongst $T_1,\ldots, T_{7n/3}$, we must have the $2n$ dummy triangles. Now, it is clear that the remaining $\frac{7n}{3}-2n=\frac{n}{3}$ triangles form a collection of pairwise vertex disjoint triangles in $G$, as desired. 
  \end{proof}
  


\section{Concluding Remarks} 

We considered the problem of modifying a graph optimally to reduce the number of distinct eigenvalues in the spectrum of its adjacency matrix. These problems turned out to be closely related to, but different from, modifications that aim to reduce the rank of the adjacency matrix and the diameter of the graph. 

The complexity of $r$-EEE for fixed $r \geqslant 3$ remains open. The parameterized complexity of $2$-EEE in the standard parameter is open, and the question of finding polynomial kernels for $2$-EVD and $2$-EED remains open as well. Studying these problems from the perspective of structural parameters or on directed graphs are interesting directions for future work.

\newpage 

\bibliography{refs}

\newpage

\end{document}